\documentclass[11pt]{article}
\usepackage[a4paper, textwidth=14cm, textheight=23cm, centering]{geometry}
\usepackage[utf8]{inputenc}
\usepackage[T1]{fontenc}
\usepackage{lmodern}
\usepackage[english]{babel}
\usepackage{setspace}
\usepackage{amsmath}
\usepackage{amssymb}
\usepackage{amsthm}
\usepackage{mathrsfs}
\usepackage{braket}
\usepackage{bbm}
\usepackage{relsize}
\usepackage{tcolorbox}
\usepackage{graphicx}
\usepackage{multicol}
\usepackage{appendix}
\usepackage{tikz-cd}
\usepackage{orcidlink}
\usepackage{hyperref}
\hypersetup{
    colorlinks,
    citecolor= orange,
    filecolor= blue,
    linkcolor= blue,
    urlcolor= green
}
\DeclareMathOperator{\Tr}{Tr}
\newcommand{\Omegahat}{\widehat{\boldsymbol{\Omega}}}
\newtheorem{definition}{Definition}
\newtheorem{lemma}{Lemma}
\newtheorem{theorem}{Theorem}
\newtheorem{corollary}{Corollary}
\newtheorem{proposition}{Proposition}
\newtheorem{remark}{Remark}

\begin{titlepage}

\title{Geometric Complexity of Quantum Channels via Unitary Dilations}
\author{Alberto Acevedo\orcidlink{https://orcid.org/0000-0001-5666-5763}, Antonio Falcó\footnote{Corresponding author: afalco@uchceu.es} \orcidlink{https://orcid.org/0000-0001-6225-0935} \\ 
Departamento de Matemáticas, Física y Ciencias Tecnológicas \\
Universidad Cardenal Herrera-CEU, CEU Universities,\\ 
Calle San Bartolome, 55, Alfara del Patriarca 46115, Valencia, Spain}
\date{}
\end{titlepage}
\begin{document}
\singlespacing

\maketitle

\vspace{0.9cm}
\begin{abstract}
Nielsen's geometric approach to quantum circuit complexity provides a Riemannian framework for quantifying the cost of implementing unitary closed-system dynamics. For open evolutions, the reduced dynamics is described by time-dependent families of quantum channels $\{\Lambda_t\}_{t\ge 0}$ admitting many inequivalent Stinespring realizations, so any meaningful complexity notion must specify which microscopic resources are counted and which transformations are treated as gauge.

We develop a dilation-based geometric complexity framework that distinguishes implementation-dependent complexity, attached to a specific realization, from intrinsic channel complexity obtained by minimizing over a physically admissible class of dilations. The proposed functional is subtractive: it removes purely environmental contributions from the geometric cost of the joint unitary realization. The subtraction is derived from concise structural postulates --- closed-system consistency, environment-only neutrality, and invariance under dilation gauge transformations --- and admits a natural variational and geometric interpretation. We also define noise complexity as a companion quantity measuring the geometric deviation of a noisy implementation from a prescribed ideal closed-system evolution.

In the Hilbert-Schmidt setting, we establish coherence-based lower bounds for unitary and noise complexity, prove linear time scaling under time-homogeneous dilation models, and obtain dissipator-controlled bounds in the Markovian GKSL regime. We illustrate the framework on canonical single-qubit benchmark channels: dephasing, amplitude damping, and Pauli/depolarizing noise.
\end{abstract}


\section{Introduction}
\label{sec:introduction}

The geometric approach to quantum circuit complexity, pioneered by
Nielsen and collaborators, provides a conceptually clean and technically
powerful framework for quantifying the cost of implementing
\emph{unitary} transformations by endowing \(SU(N)\) with a
right-invariant metric that encodes available controls and penalties
\cite{Niel,dorth1,dorth2,brandt}. In the closed-system setting, this viewpoint yields
a Riemannian (or Finsler) length functional whose minimizers admit a
transparent variational interpretation and whose value is relatable to gate-counting notions of complexity
\cite{Niel}.

Extending this picture to open quantum dynamics is, however,
intrinsically non-unique. A general noisy evolution of a system is
described by a family of quantum channels \(\{\Lambda_t\}_{t\ge 0}\),
and such a family is typically \emph{time-dependent} and, in the non-Markovian regime,
need not form a semigroup \cite{BreuerPetruccione,Schlosshauer,Zurek}.
Moreover, any such family may be realized through many inequivalent
microscopic implementations, for instance via Stinespring dilations that
differ in environment dimension, reservoir preparation, or
system-environment coupling \cite{BreuerPetruccione, Stinespring1955}. Consequently, the notion of ``accessible cost'' for open dynamics cannot be interpreted solely from the reduced channel: it depends on which degrees
of freedom are regarded as physical resources and which are treated as
a gauge. This creates a conceptual gap between the well-established
geometric complexity of closed dynamics and a meaningful complexity
theory for noisy processes which are modeled via the Stinespring dilation approach.

The present work addresses this gap by proposing and analyzing a
geometric complexity framework for open-system evolutions based on
unitary dilations. The starting point is the distinction between two layers. The first is an \emph{implementation-dependent} complexity, where the microscopic dilation data is regarded as part of the physical
description and cost is assigned to a specific realization of the
channel. The second is an \emph{intrinsic} channel complexity, obtained
by minimizing over a physically admissible class of dilations
\cite{Burgarth,Cramer}. This
distinction is essential: without it, one risks conflating implementations of the same channel which are physically different, or else assigning cost in
a way that is insensitive to the actual microscopic resources involved.

Within this two-layer framework, we introduce a dilation-based channel
complexity functional
\(\mathcal G_{\widehat{\Omega}}(\Lambda_t;\mathfrak D)\) built from the geometric
cost of the joint unitary realization and a canonical subtraction of the
purely environmental contribution. The subtraction is not introduced ad hoc. Rather, it is derived from a compact set of physically motivated structural postulates: consistency with the closed-system limit, neutrality under purely environmental evolution, and invariance under the natural dilation
gauge given by environmental changes of basis that leave the reduced
channel unchanged. We also provide a variational and geometric interpretation of the subtraction, clarifying why it is the natural extension of the closed-system geometry within a dilation-aware model. The resulting functional retains the geometric spirit
of Nielsen complexity while removing contributions that are invisible at
the level of the reduced dynamics. We further introduce a related quantity, which we call \emph{noise complexity}, which measures
the geometric deviation between a noisy implementation and a prescribed
ideal closed-system evolution. In this
sense, channel complexity and noise complexity separate two distinct but
related questions: how costly a noisy implementation is in its own
right, and how far that implementation lies from a chosen ideal
implementation.

A central feature of the present framework is the fact that the geometry itself
carries resource information. The penalty tensor $\Omegahat$ (characterizing the non-homogeneous Riemannian metric), is not used
merely as a mathematical convenience: it assigns different costs to
different operator directions, encoding which controls or
couplings are regarded as expensive and which are not. Intrinsic complexity is thus
resource-relative in two senses: first through the admissible class of
dilations, and second through the choice of geometric weighting
via the matrix \(\Omegahat\). This makes explicit that there is no absolute notion of
channel complexity divorced from implementation assumptions.

On the technical side, we obtain several results clarifying the
structure and interpretation of the functionals hitherto described. First, we establish a
coherence-based lower bound for unitary geometric complexity in the
Hilbert-Schmidt specialization. This gives an operational meaning to
the geometric cost by relating it to the cohering power of a
unitary evolution relative to a fixed reference basis. Second, we extend this
perspective to noise complexity by applying the same coherence-based
logic to the relative noise unitary operator, thereby obtaining lower bounds that
complement the more direct deficit-type estimates derived later in the
weak-coupling regime. Third, we prove structural properties of the
channel functional, including linear time-scaling under
time-homogeneous dilation models and explicit dissipator-controlled
bounds in the Markovian (GKSL/Lindblad) setting \cite{BreuerPetruccione,Lindblad1976,GKS1976}.

To connect the abstract framework with standard open-system models, we
then study canonical single-qubit channels in a
Davies-motivated weak-coupling regime. These examples should not be read
as exact microscopic reconstructions of the channel geometry, but rather
as benchmark realizations within a physically motivated admissible class
of dilations. They illustrate how dephasing, amplitude damping, and
Pauli/depolarizing noise are reflected in the implementation-dependent
channel complexity and in the associated noise complexity.

This paper is aimed at three overlapping audiences. For researchers in
quantum information and quantum control, it provides a concrete proposal
for quantifying the cost of noisy implementations beyond gate counting.
For mathematical physicists, it offers an explicit geometric functional
with a clear invariance structure and a resource-sensitive dependence on
the chosen penalty geometry. For the open quantum systems community, it
provides a bridge between weak-coupling Markovian models and geometric
cost estimates for noisy dynamics.

\paragraph{Contributions.}
The main contributions of this work can be summarized as follows:
\begin{itemize}
  \item We introduce a dilation-aware geometric complexity functional
  $\mathcal{G}_{\Omegahat}(\Lambda_t;\mathfrak{D})$ for time-dependent families of quantum channels
  $\{\Lambda_t\}_{t\ge 0}$, explicitly separating \emph{implementation-dependent} costs from
  \emph{intrinsic} channel complexity obtained by minimization over a physically motivated class of
  admissible dilations (resource constraints such as bounded environment dimension, energy/norm
  bounds, and penalty structures).
  \item We propose a \emph{subtractive} definition that removes purely environmental (system-invisible)
  contributions and justify the subtraction from concise postulates: closed-system consistency,
  environment-only neutrality, and invariance under dilation gauge transformations that leave the
  reduced channel unchanged.
  \item We define \emph{noise complexity} as a companion quantity quantifying complexity loss
  relative to a prescribed ideal closed evolution.
  \item We prove an operational lower bound on unitary geometric complexity in terms of
  basis-dependent coherence production \cite{Baumgratz2014}, providing a physically meaningful control parameter for the
  geometric cost.
  \item We derive structural properties of the channel complexity, including linear time-scaling
  under time-homogeneous dilation models, and obtain dissipator-controlled bounds in the Markovian
  (GKSL/Lindblad) regime under standard dilation constructions.
  \item We illustrate the framework on canonical benchmark channels (dephasing, amplitude damping,
  and depolarizing/Pauli noise), and interpret the resulting complexity and noise-complexity trends
  \cite{suss1,suss2,SusskindLectures2018}.
\end{itemize}

The paper is organized as follows. Section~\ref{sec:prelim} fixes
notation and conventions. Section~\ref{sec:unitary_complexity} recalls
unitary geometric complexity in the Nielsen framework.
Section~\ref{sec:defs-channel-noise} introduces the implementation-dependent and
intrinsic channel complexities, together with the respective noise complexities and the
associated notion of resource structures. Section~\ref{sec:axioms-subtractive}
states the postulates and derives the subtractive form of the channel complexity. Section~\ref{sec:coherence-bounds} develops the
coherence-based lower bounds for unitary and noise complexity.
Section~\ref{sec:lindblad} treats the Markovian weak-coupling regime.
Finally, Section~\ref{sec:benchmarks} presents canonical benchmark
examples, and Section~\ref{sec:conclusions} concludes with final
remarks.

\section{Notations and Conventions}
\label{sec:prelim}

This section fixes notation and standard conventions used throughout the paper. Unless explicitly
stated otherwise, all Hilbert spaces are finite-dimensional and over $\mathbb{C}$, and we set
$\hbar = 1$.

\subsection{Hilbert spaces, operators, and norms}

We denote by $\mathscr{H}$ a finite-dimensional Hilbert space of dimension $\dim(\mathscr{H})$.
The space of linear operators on $\mathscr{H}$ is denoted by
$\mathcal{L}(\mathscr{H})$.
The identity operator on $\mathscr{H}$ is denoted by $\mathbf{\hat{I}}_{\mathscr{H}}$
(or simply $\mathbf{\hat{I}}$ when the underlying space is clear).
Adjoints are denoted by a dagger superscript ${}^\dagger$.

The set of density operators (quantum states) on $\mathscr{H}$ is
\begin{equation}
\mathcal{S}(\mathscr{H})
:=
\big\{
\boldsymbol{\hat{\rho}} \in \mathcal{L}(\mathscr{H})
:\;
\boldsymbol{\hat{\rho}} \succeq 0,\;
\operatorname{Tr}(\boldsymbol{\hat{\rho}})=1
\big\}.
\end{equation}
We write $\operatorname{Tr}$ for the trace, and $\operatorname{Tr}_{E}$ for the partial trace over environmental degrees of freedom $\mathscr{H}_E$ (see below). Commutators are written
$[\mathbf{\hat{A}},\mathbf{\hat{B}}]=\mathbf{\hat{A}}\mathbf{\hat{B}}-\mathbf{\hat{B}}\mathbf{\hat{A}}$.

We will denote the Hilbert-Schmidt inner product as follows:
\begin{equation}
\langle \mathbf{\hat{A}},\mathbf{\hat{B}} \rangle_{HS}
:=\operatorname{Tr}(\mathbf{\hat{A}}^\dagger \mathbf{\hat{B}}),
\qquad
\|\mathbf{\hat{A}}\|_{HS}
:=\sqrt{\langle \mathbf{\hat{A}},\mathbf{\hat{A}} \rangle_{HS}},
\end{equation}
and, when needed, the trace norm $\|\mathbf{\hat{A}}\|_{1}:=\operatorname{Tr}\sqrt{\mathbf{\hat{A}}^\dagger\mathbf{\hat{A}}}$.
All logarithms in entropic quantities are natural logarithms unless indicated otherwise.

\subsection{System-environment split and unitary dynamics}
Open-system dynamics are formulated on a tensor-product Hilbert space $\mathscr{H}=\mathscr{H}_{S}\otimes \mathscr{H}_{E}$. Furthermore, states of the system $S$ and environment $E$ are $\boldsymbol{\hat{\rho}}_S\in\mathcal{S}(\mathscr{H}_S)$ and
$\boldsymbol{\hat{\rho}}_E\in\mathcal{S}(\mathscr{H}_E)$, respectively, and it is assumed that the initial state is uncorrelated. Namely, 
\begin{equation}
\label{eqn:initial_product_state}
\boldsymbol{\hat{\rho}}_{0}
=
\boldsymbol{\hat{\rho}}_{S}\otimes \boldsymbol{\hat{\rho}}_{E}.
\end{equation}
For a time-independent Hamiltonian $\boldsymbol{\hat{H}}_{tot}$ on $\mathscr{H}$, the total unitary
evolution generated by said Hamiltonian is $\mathbf{\hat{U}}_{t}=e^{-it\boldsymbol{\hat{H}}_{tot}}$.
When considering a time-dependent Hamiltonian $\boldsymbol{\hat{H}}_{tot}(t)$, the corresponding unitary is
\begin{equation}
\label{eqn:time_ordered_unitary_prelim}
\mathbf{\hat{U}}_{t}
=
\hat{\mathscr{T}}\exp\!\left(-i\int_{0}^{t}\boldsymbol{\hat{H}}_{tot}(s)\,ds\right),
\end{equation}
where $\hat{\mathscr{T}}$ denotes the time-ordering operator.

Whenever convenient, we use the standard decomposition
\begin{equation}
\label{eqn:hamiltonian_split_prelim}
\boldsymbol{\hat{H}}_{tot}(s)
=
\boldsymbol{\hat{H}}_{S}(s)
+
\boldsymbol{\hat{H}}_{I}(s)
+
\boldsymbol{\hat{H}}_{E}(s),
\end{equation}
where $\boldsymbol{\hat{H}}_{S}(s)$ acts non-trivially only on $\mathscr{H}_{S}$,
$\boldsymbol{\hat{H}}_{E}(s)$ acts non-trivially only on $\mathscr{H}_{E}$, and the interaction
$\boldsymbol{\hat{H}}_{I}(s)$ acts on $\mathscr{H}_{S}\otimes\mathscr{H}_{E}$ for all $0\leq s \leq t$.

\subsection{Quantum channels and Kraus/Stinespring representations}

A (finite-dimensional) quantum channel on $\mathscr{H}_S$ is a linear map
$\Lambda:\mathcal{L}(\mathscr{H}_S)\to\mathcal{L}(\mathscr{H}_S)$ that is completely positive and
trace-preserving (CPTP). The simplest example of these are unitary channels which have the form
$\Lambda(\boldsymbol{\hat{\rho}})=\mathbf{\hat{U}}\boldsymbol{\hat{\rho}}\mathbf{\hat{U}}^\dagger$.

More generally, every CPTP map admits a Kraus representation
\begin{equation}
\label{eqn:kraus_prelim}
\Lambda(\boldsymbol{\hat{\rho}})
=
\sum_{\alpha} \mathbf{\hat{K}}_{\alpha}\,\boldsymbol{\hat{\rho}}\,\mathbf{\hat{K}}_{\alpha}^\dagger,
\qquad
\sum_{\alpha}\mathbf{\hat{K}}_{\alpha}^\dagger \mathbf{\hat{K}}_{\alpha}
=
\mathbf{\hat{I}}_{\mathscr{H}_S}.
\end{equation}
Equivalently, every CPTP map admits a Stinespring dilation \cite{Stinespring1955}; i.e. there exist an environmental Hilbert space
$\mathscr{H}_E$, an environmental state $\boldsymbol{\hat{\rho}}_E\in\mathcal{S}(\mathscr{H}_E)$, and a unitary operator
$\mathbf{\hat{U}}$ acting on $\mathscr{H}_S\otimes\mathscr{H}_E$ such that
\begin{equation}
\label{eqn:stinespring_prelim}
\Lambda(\boldsymbol{\hat{\rho}}_S)
=
\operatorname{Tr}_E\!\big[
\mathbf{\hat{U}}\,
(\boldsymbol{\hat{\rho}}_S\otimes \boldsymbol{\hat{\rho}}_E)\,
\mathbf{\hat{U}}^\dagger
\big].
\end{equation}
In this paper we consider time-parameterized channels $\{\Lambda_t\}_{t\ge 0}$ obtained by taking
$\mathbf{\hat{U}}=\mathbf{\hat{U}}_t$ from \eqref{eqn:time_ordered_unitary_prelim} (or
$\mathbf{\hat{U}}_t=e^{-it\boldsymbol{\hat{H}}_{tot}}$ in the time-independent case), together with
a fixed choice of an initial environmental state $\boldsymbol{\hat{\rho}}_E$.

In several parts of the paper (in particular, when defining noise complexity) we compare the reduced
open dynamics to the ideal closed evolution generated solely by a family of Hermitian matrices $\boldsymbol{\hat{H}}_{S}(s)$. We thus
introduce the corresponding unitary on $\mathscr{H}_S$,
\begin{equation}
\label{eq:ideal-system-unitary}
\mathbf{\hat{U}}_{S}(t)
=
\hat{\mathscr{T}}\exp\!\left(-i\int_{0}^{t}\boldsymbol{\hat{H}}_{S}(s)\,ds\right)
\end{equation}
and the associated unitary channel $\mathcal{U}_S(t)$ defined by
$\mathcal{U}_S(t)(\boldsymbol{\hat{\rho}}):=\boldsymbol{\hat{U}}_S(t)\boldsymbol{\hat{\rho}}\,\boldsymbol{\hat{U}}_S(t)^\dagger$.

\subsection{Dilation ``gauge'' freedom and implementation dependence}

A key structural point is that the dilation Eq.\eqref{eqn:stinespring_prelim} is not unique:
different triples $(\mathscr{H}_E,\boldsymbol{\hat{\rho}}_E,\boldsymbol{\hat U})$ may induce the same reduced
channel $\Lambda$. For instance, if $\boldsymbol{\hat W}$ is a unitary operator acting on
$\mathscr{H}_E$, and we define
\begin{equation}
\boldsymbol{\hat U}'
=
(\boldsymbol{\hat I}_S\otimes\boldsymbol{\hat W})\,
\boldsymbol{\hat U}\,
(\boldsymbol{\hat I}_S\otimes\boldsymbol{\hat W}^\dagger),
\qquad
\boldsymbol{\hat\rho}'_E
=
\boldsymbol{\hat W}\boldsymbol{\hat\rho}_E\boldsymbol{\hat W}^\dagger,
\end{equation}
then
\begin{equation}
\operatorname{Tr}_E\!\big[
\boldsymbol{\hat U}'
(\boldsymbol{\hat\rho}_S\otimes\boldsymbol{\hat\rho}'_E)
\boldsymbol{\hat U}'^\dagger
\big]
=
\operatorname{Tr}_E\!\big[
\boldsymbol{\hat U}
(\boldsymbol{\hat\rho}_S\otimes\boldsymbol{\hat\rho}_E)
\boldsymbol{\hat U}^\dagger
\big].
\end{equation}
Thus a unitary change of description on the environment can be absorbed into a simultaneous
redefinition of both the environment state and the dilation unitary, without changing the
induced channel on the system. More generally, enlargements of $\mathscr{H}_E$ and different
Stinespring isometries can also lead to the same channel.
\medskip

\noindent\textbf{Convention (implementation model):}
Whenever we speak of ``the'' open-dynamics $\Lambda_t$ induced by a physical noise mechanism, we
tacitly assume that a specific microscopic implementation, namely a choice of
$(\mathscr{H}_E,\boldsymbol{\hat{\rho}}_E,\boldsymbol{\hat{H}}_{tot})$ (or $(\mathscr{H}_E,\boldsymbol{\hat{\rho}}_E,\boldsymbol{\hat{H}}_{tot}(t))$ for the time-dependent case), is fixed.
Consequently, any complexity functional built from dilation data should be interpreted as an
\emph{implementation-dependent} cost, unless an explicit optimization over dilations is performed.
To keep notation light we typically suppress this dependence, but it can be restored by writing
$\Lambda_t^{(\mathscr{H}_E,\boldsymbol{\hat{\rho}}_E,\boldsymbol{\hat{H}}_{tot})}$ when needed.

\subsection*{Summary of notation}

For the reader's convenience, the following table collects the principal symbols used throughout
the paper.

\medskip
\noindent
\begin{tabular}{lp{9cm}}
\hline
\textbf{Symbol} & \textbf{Meaning} \\
\hline
$\mathscr{H}_S$, $\mathscr{H}_E$, $\mathscr{H}$ & System, environment, and total Hilbert spaces; $\mathscr{H}=\mathscr{H}_S\otimes\mathscr{H}_E$ \\
$d_S$, $d_E$, $d_{\mathrm{tot}}$ & Dimensions of $\mathscr{H}_S$, $\mathscr{H}_E$, and $\mathscr{H}$ \\
$\mathcal{S}(\mathscr{H})$ & Set of density operators on $\mathscr{H}$ \\
$\mathcal{L}(\mathscr{H})$ & Space of linear operators on $\mathscr{H}$ \\
$\langle\cdot,\cdot\rangle_{HS}$, $\|\cdot\|_{HS}$ & Hilbert-Schmidt inner product and norm: $\langle A,B\rangle_{HS}=\operatorname{Tr}(A^\dagger B)$ \\
$\langle\cdot,\cdot\rangle_{\Omegahat}$, $\|\cdot\|_{\Omegahat}$ & $\Omegahat$-weighted inner product and norm on $\mathfrak{su}(N)$ (Definition~\ref{eq:Omega_inner_product}) \\
$\Omegahat$ & Penalty matrix (Hermitian positive definite diagonal); encodes geometric resource weights \\
$\mathfrak{su}(N)$ & Lie algebra of $SU(N)$; $N=d_{\mathrm{tot}}$ \\
$\mathscr{G}_{\Omegahat}(\boldsymbol{\hat{U}})$ & Unitary geometric complexity of $\boldsymbol{\hat{U}}$ (Definition~\ref{def:unitary_complexity}) \\
$\mathscr{D}_{\Omegahat}(\boldsymbol{\hat{U}},\boldsymbol{\hat{V}})$ & Carnot--Carath\'eodory distance induced by $\Omegahat$ \\
$\mathfrak{D}=(\mathscr{H}_E,\boldsymbol{\hat{\rho}}_E,\boldsymbol{\hat{H}}_{\mathrm{tot}})$ & Stinespring dilation data \\
$\Lambda_t$, $\Lambda_t^{(\mathfrak{D})}$ & Quantum channel family; channel induced by dilation $\mathfrak{D}$ \\
$\boldsymbol{\hat{U}}_S(t)$ & Ideal closed-system unitary; $\boldsymbol{\hat{U}}_S(t)=\hat{\mathscr{T}}\exp(-i\int_0^t\boldsymbol{\hat{H}}_S(s)\,ds)$ \\
$\mathcal{G}^{\mathrm{coh}}_{\Omegahat}(\Lambda_t;\mathfrak{D})$ & Coherent distortion complexity (Definition~\ref{eq:coherent-distortion-complexity}) \\
$\mathcal{G}^{\mathrm{imp}}_{\Omegahat}(\Lambda_t;\mathfrak{D})$ & Implementation-dependent channel complexity (Definition~\ref{eq:implementation-dependent-channel-complexity}) \\
$\mathcal{N}_{\Omegahat}(\Lambda_t;\mathfrak{D})$ & Geometric noise complexity (Definition~\ref{eqn:noisecomplexity}) \\
$\mathcal{G}^{\mathcal{R}}_{\mathrm{intr}}(\Lambda_t;[0,T])$ & Intrinsic channel complexity relative to resource model $\mathcal{R}$ (Definition~\ref{def:intrinsic-channel-complexity}) \\
$\boldsymbol{\hat{P}}_S$, $\boldsymbol{\hat{P}}_{int}$, $\boldsymbol{\hat{P}}_E$ & Orthogonal projectors onto system-local, interaction, and environmental subspaces of $\mathfrak{su}(d_Sd_E)$ \\
$\mathfrak{Dil}_{\mathcal{R}}^{[0,T]}(\Lambda_t)$ & Set of admissible dilations of $\Lambda_t$ on $[0,T]$ under resource model $\mathcal{R}$ \\
$\Gamma=\sum_\alpha\|\boldsymbol{\hat{L}}_\alpha\|_{\mathrm{op}}^2$ & Dissipator scale for a GKSL generator with Lindblad operators $\boldsymbol{\hat{L}}_\alpha$ \\
$C_{\mathscr{E}}(\boldsymbol{\hat{\rho}})$ & Linear-entropy coherence of $\boldsymbol{\hat{\rho}}$ relative to dephasing map $\mathscr{E}$ \\
$\mathscr{C}_{\mathscr{E}}(\boldsymbol{\hat{U}})$ & Cohering power of unitary $\boldsymbol{\hat{U}}$ (Definition~\ref{eqn:coheringpower}) \\
$\lambda$ & Weak-coupling parameter; $0<\lambda\ll 1$ in the Davies surrogate \\
\hline
\end{tabular}
\medskip

\section{Unitary Geometric Complexity}
\label{sec:unitary_complexity}

This section recalls the geometric approach to unitary circuit complexity initiated by Nielsen and
coauthors \cite{Niel} and developed further in, e.g., \cite{dorth1,dorth2,brandt}.  The guiding
principle is to encode physical constraints (locality, available interactions, control costs) into a
right--invariant Riemannian metric on $SU(N)$, and to define the complexity of a target unitary as the
length of the shortest admissible path from the identity to that unitary.

\subsection{Right--invariant Riemannian metrics on $SU(N)$}

Let
\begin{equation}
SU(N)=\big\{\boldsymbol{\hat{U}}\in\mathbb{C}^{N\times N}:\ \boldsymbol{\hat{U}}^\dagger \boldsymbol{\hat{U}}=\boldsymbol{\hat{I}}_N,\ \det(\boldsymbol{\hat{U}})=1\big\},
\end{equation}
a compact Lie group of real dimension $N^2-1$. Its Lie algebra is
\begin{equation}
\mathfrak{su}(N)=\big\{\boldsymbol{\hat{X}}\in\mathbb{C}^{N\times N}:\ \boldsymbol{\hat{X}}^\dagger=-\boldsymbol{\hat{X}},\ \mathrm{Tr}(\boldsymbol{\hat{X}})=0\big\}.
\end{equation}
Fix a real basis $\{\hat{\mathbf{E}}_i\}_{i=1}^{N^2-1}$ of $\mathfrak{su}(N)$. Every
$\hat{\mathbf{A}}\in\mathfrak{su}(N)$ admits a unique expansion
\begin{equation}
\hat{\mathbf{A}}=\sum_{i=1}^{N^2-1} A_i\,\hat{\mathbf{E}}_i,\qquad A_i\in\mathbb{R},
\end{equation}
and we define the associated coordinate (vectorization) map
\begin{equation}
\mathrm{Vec}(\hat{\mathbf{A}})
:=
\begin{pmatrix}
A_1\\
\vdots\\
A_{N^2-1}
\end{pmatrix}
\in\mathbb{R}^{N^2-1}.
\end{equation}

Let $\Omegahat\in\mathbb{C}^{(N^2-1)\times(N^2-1)}$ be Hermitian positive definite. We define an inner
product on $\mathfrak{su}(N)$ by
\begin{equation}
\label{eq:Omega_inner_product}
\langle \hat{\mathbf{A}},\hat{\mathbf{B}}\rangle_{\Omegahat}
:=
\frac{1}{\operatorname{Tr}(\Omegahat)}\,
\mathrm{Vec}(\hat{\mathbf{A}})^\dagger\,
\Omegahat\,
\mathrm{Vec}(\hat{\mathbf{B}}),
\qquad
\hat{\mathbf{A}},\hat{\mathbf{B}}\in\mathfrak{su}(N).
\end{equation}

\paragraph{Normalization.}
The $\Omegahat$-weighted geometry is defined only up to an overall scaling of the metric, i.e.\ under
transformations $\Omegahat \mapsto c\,\Omegahat$ with $c>0$. To remove this ambiguity, we fix the
normalization by requiring that the overall scale of the geometry be set by $\operatorname{Tr}(\Omegahat)$,
the total weight assigned by $\Omegahat$ to the basis directions. This choice ensures that the resulting
geometry is invariant under global rescalings of $\Omegahat$ while preserving its directional weighting.
In particular, it provides a dimension-independent and resource-consistent normalization, reducing to
$1/(N^2-1)$ in the Hilbert-Schmidt case $\Omegahat = \boldsymbol{\hat{I}}_{N^2-1}$.

We extend \eqref{eq:Omega_inner_product} to a right--invariant Riemannian metric on $SU(N)$ as follows.
For $\boldsymbol{\hat{U}}\in SU(N)$, each tangent vector $\boldsymbol{\hat{X}}\in T_{\boldsymbol{\hat{U}}} SU(N)$
can be written uniquely as $\boldsymbol{\hat{X}}=\hat{\mathbf{A}}\,\boldsymbol{\hat{U}}$ with
$\hat{\mathbf{A}}\in\mathfrak{su}(N)$ (right trivialization). We set
\begin{equation}
\label{eq:right_invariant_metric}
g^{(\Omegahat)}_{\boldsymbol{\hat{U}}}(\boldsymbol{\hat{X}},\boldsymbol{\hat{Y}})
:=
\big\langle \hat{\mathbf{A}},\hat{\mathbf{B}}\big\rangle_{\Omegahat}
\qquad
\text{if } \boldsymbol{\hat{X}}=\hat{\mathbf{A}}\boldsymbol{\hat{U}},\ \boldsymbol{\hat{Y}}=\hat{\mathbf{B}}\boldsymbol{\hat{U}}.
\end{equation}
Right invariance is immediate: for every $\boldsymbol{\hat{W}}\in SU(N)$,
$g^{(\Omegahat)}_{\boldsymbol{\hat{U}}\boldsymbol{\hat{W}}}(\boldsymbol{\hat{X}}\boldsymbol{\hat{W}},\boldsymbol{\hat{Y}}\boldsymbol{\hat{W}})=g^{(\Omegahat)}_{\boldsymbol{\hat{U}}}(\boldsymbol{\hat{X}},\boldsymbol{\hat{Y}})$ \cite{lee}.

Throughout, we will often assume that the chosen basis diagonalizes $\Omegahat$,
\begin{equation}
\label{eqn:metric}
\Omegahat=\mathrm{diag}(l_1,\ldots,l_{N^2-1}),\qquad l_i\geq 1,
\end{equation}
which is natural in applications where one assigns explicit ``penalties'' to specific interaction
directions, a value of 1 meaning no penalty while a value greater than one corresponds to a penalty.
\subsection{Homogeneous versus anisotropic geometries}

If $\Omegahat$ is a scalar multiple of the identity, then $g^{(\Omegahat)}$ is bi--invariant (equivalently,
adjoint--invariant) and $(SU(N),g^{(\Omegahat)})$ is a homogeneous Riemannian manifold. In that case,
geodesics are one--parameter subgroups,
\begin{equation}
\gamma(s)=\exp\!\big(s\,\hat{\mathbf{A}}\big),
\qquad \hat{\mathbf{A}}\in\mathfrak{su}(N),
\end{equation}
and distances are controlled by minimal--norm logarithms.

If $\Omegahat$ is not proportional to the identity, the metric is typically only right--invariant. The
anisotropy breaks adjoint symmetry and geodesics need not be one--parameter subgroups; instead, the
right--trivialized velocity satisfies an Euler--Arnold equation (see
Section~\ref{subsec:euler_arnold}).

\subsection{Example: penalizing nonlocal interactions (multi--qubit setting)}

For an $n$--qubit register, $N=2^n$ and $\dim\mathfrak{su}(2^n)=4^n-1$.  Let
$\mathfrak{g}_{\mathrm{loc}}$ be the span of $1$-- and $2$--local Pauli strings and
$\mathfrak{g}_{\mathrm{nl}}$ its orthogonal complement (with respect to the Hilbert-Schmidt product),
so that
\begin{equation}
\mathfrak{su}(2^n)=\mathfrak{g}_{\mathrm{loc}}\oplus \mathfrak{g}_{\mathrm{nl}}.
\end{equation}
A standard locality penalty is realized by choosing a basis adapted to this splitting and setting
\begin{equation}
\label{eq:penalty_matrix}
\Omegahat_q
=
\boldsymbol{\hat{I}}_{\dim \mathfrak{g}_{\mathrm{loc}}}
\ \oplus\
q\,\boldsymbol{\hat{I}}_{\dim \mathfrak{g}_{\mathrm{nl}}},
\qquad q>1.
\end{equation}
Then directions corresponding to nonlocal (higher--body) interactions are ``stretched'' by a factor
$q$, and minimal--length paths preferentially use inexpensive local directions. This is precisely the
geometric counterpart of introducing gate penalties in circuit models \cite{Niel,dorth2}.

\subsection{Length, distance, and unitary geometric complexity}

Let $\gamma:[0,1]\to SU(N)$ be piecewise $C^1$. Its length with respect to $g^{(\Omegahat)}$ is
\begin{equation}
\label{eq:length}
L_{\Omegahat}(\gamma)
:=
\int_0^1 \sqrt{\,g^{(\Omegahat)}_{\gamma(s)}\!\big(\dot{\gamma}(s),\dot{\gamma}(s)\big)\,}\,ds.
\end{equation}
The induced (Riemannian, or Carnot--Carath\'eodory) distance between $\boldsymbol{\hat{U}},\boldsymbol{\hat{V}}\in SU(N)$ is
\begin{equation}
\label{eq:distance}
\mathscr{D}_{\Omegahat}(\boldsymbol{\hat{U}},\boldsymbol{\hat{V}})
:=
\inf\Big\{L_{\Omegahat}(\gamma):\ \gamma(0)=\boldsymbol{\hat{U}},\ \gamma(1)=\boldsymbol{\hat{V}}\Big\}.
\end{equation}
Since $SU(N)$ is compact and $g^{(\Omegahat)}$ is smooth, $(SU(N),g^{(\Omegahat)})$ is geodesically complete
and the infimum in \eqref{eq:distance} is attained by at least one minimizing geodesic (Hopf--Rinow) \cite{lee}.

\begin{definition}[Geometric complexity of a unitary]
\label{def:unitary_complexity}
For $\boldsymbol{\hat{U}}\in SU(N)$, the (unitary) geometric complexity induced by $\Omegahat$ is
\begin{equation}
\label{eq:unitary_complexity}
\mathscr{G}_{\Omegahat}(\boldsymbol{\hat{U}})
:=
\mathscr{D}_{\Omegahat}(\boldsymbol{\hat{I}}_N,\boldsymbol{\hat{U}}).
\end{equation}
\end{definition}

Right invariance of the metric implies right invariance of the distance:
\begin{equation}
\label{eq:right_invariance_distance}
\mathscr{D}_{\Omegahat}(\boldsymbol{\hat{U}},\boldsymbol{\hat{V}})=\mathscr{D}_{\Omegahat}(\boldsymbol{\hat{I}}_N,\,\boldsymbol{\hat{V}}\boldsymbol{\hat{U}}^{\dagger})
\qquad\text{for all }\boldsymbol{\hat{U}},\boldsymbol{\hat{V}}\in SU(N).
\end{equation}
In particular, $\mathscr{G}_{\Omegahat}(\boldsymbol{\hat{U}})$ measures the minimal cost of synthesizing
$\boldsymbol{\hat{U}}$ from the identity under the geometric cost encoded by $\Omegahat$.

\subsection{Control (Hamiltonian) representation of the length functional}
\label{subsec:control_representation}

For a smooth curve $\gamma(s)\in SU(N)$, define its right--trivialized velocity (``body velocity'')
\begin{equation}
\label{eq:body_velocity}
\hat{\mathbf{A}}(s)
:=
\dot{\gamma}(s)\,\gamma(s)^{-1}
=
\dot{\gamma}(s)\,\gamma(s)^\dagger
\in\mathfrak{su}(N).
\end{equation}
Then $\dot{\gamma}(s)=\hat{\mathbf{A}}(s)\gamma(s)$, and by construction of the metric,
\begin{equation}
g^{(\Omegahat)}_{\gamma(s)}\!\big(\dot{\gamma}(s),\dot{\gamma}(s)\big)
=
\big\langle \hat{\mathbf{A}}(s),\hat{\mathbf{A}}(s)\big\rangle_{\Omegahat}.
\end{equation}
Hence the length \eqref{eq:length} becomes
\begin{equation}
\label{eq:length_body_velocity}
L_{\Omegahat}(\gamma)
=
\int_0^1 \sqrt{\big\langle \hat{\mathbf{A}}(s),\hat{\mathbf{A}}(s)\big\rangle_{\Omegahat}}\,ds.
\end{equation}

It is often convenient to switch to Hermitian Hamiltonians.  Using the identification
$\mathfrak{su}(N)=\{-i\boldsymbol{\hat{H}}:\ \boldsymbol{\hat{H}}=\boldsymbol{\hat{H}}^\dagger,\ \mathrm{Tr}(\boldsymbol{\hat{H}})=0\}$, set
\begin{equation}
\label{eq:Hamiltonian_identification}
\hat{\mathbf{A}}(s)=-i\,\hat{\mathbf{H}}(s),
\qquad
\hat{\mathbf{H}}(s)=\hat{\mathbf{H}}(s)^\dagger,\ \mathrm{Tr}\,\hat{\mathbf{H}}(s)=0.
\end{equation}
Then $\gamma$ solves the time--dependent Schr\"odinger equation
\begin{equation}
\label{eq:Schrodinger}
\dot{\gamma}(s)=-i\,\hat{\mathbf{H}}(s)\gamma(s),
\qquad
\gamma(0)=\boldsymbol{\hat{I}}_N,
\end{equation}
with formal solution
\begin{equation}
\label{eq:time_ordered}
\gamma(s)=\hat{\mathscr{T}}\exp\!\left(-i\int_0^s \hat{\mathbf{H}}(u)\,du\right),
\end{equation}
where $\hat{\mathscr{T}}$ denotes time ordering. In terms of $\hat{\mathbf{H}}(s)$,
\eqref{eq:length_body_velocity} reads
\begin{equation}
\label{eq:length_H}
L_{\Omegahat}(\gamma)
=
\int_0^1 \sqrt{\big\langle -i\hat{\mathbf{H}}(s),-i\hat{\mathbf{H}}(s)\big\rangle_{\Omegahat}}\,ds
=
\int_0^1 \sqrt{\big\langle \hat{\mathbf{H}}(s),\hat{\mathbf{H}}(s)\big\rangle_{\Omegahat}}\,ds,
\end{equation}
where in the last equality we use that $\langle\cdot,\cdot\rangle_{\Omegahat}$ is defined on
$\mathfrak{su}(N)$ and extended by the identification \eqref{eq:Hamiltonian_identification}.

Consequently, Definition~\ref{def:unitary_complexity} is equivalently the optimal--control problem
\begin{equation}
\label{eq:control_formulation}
\mathscr{G}_{\Omegahat}(U)
=
\inf_{\hat{\mathbf{H}}(\cdot)}
\left\{
\int_0^1 \sqrt{\big\langle \hat{\mathbf{H}}(s),\hat{\mathbf{H}}(s)\big\rangle_{\Omegahat}}\,ds
\ :\
\gamma(1)=\mathbf{\hat{U}}\ \text{with }\gamma \text{ solving \eqref{eq:Schrodinger}}
\right\}.
\end{equation}
This is the precise sense in which the metric $\Omegahat$ encodes the cost of Hamiltonian directions.

\subsection{Geodesic equation (Euler--Arnold form)}
\label{subsec:euler_arnold}

For completeness, we record the intrinsic geodesic equation associated with a right--invariant
metric on a Lie group, in the Euler--Arnold (or Euler--Poincar\'e) form
\cite{,Milnor1976,MarsdenRatiu1999,Arnold1966}.
Let $\hat{\mathbf{A}}(s)=\dot{\gamma}(s)\gamma(s)^{-1}$ be the body velocity
\eqref{eq:body_velocity}. Define the inertia operator
$\mathcal{I}_{\Omegahat}:\mathfrak{su}(N)\to\mathfrak{su}(N)$ by requiring
\begin{equation}
\label{eq:inertia_operator}
\langle \hat{\mathbf{A}},\hat{\mathbf{B}}\rangle_{\Omegahat}
=
\langle \mathcal{I}_{\Omegahat} \hat{\mathbf{A}},\hat{\mathbf{B}}\rangle_{\mathrm{HS}},
\end{equation}
where $\langle\cdot,\cdot\rangle_{\mathrm{HS}}$ denotes the Hilbert-Schmidt pairing
(any fixed nondegenerate pairing suffices).
Then $\gamma$ is a geodesic if and only if $\hat{\mathbf{A}}(s)$ satisfies the Euler--Arnold equation
\begin{equation}
\label{eq:euler_arnold}
\frac{d}{ds}\big(\mathcal{I}_{\Omegahat} \hat{\mathbf{A}}(s)\big)
=
\mathrm{ad}^*_{\hat{\mathbf{A}}(s)}\big(\mathcal{I}_{\Omegahat} \hat{\mathbf{A}}(s)\big),
\end{equation}
where $\mathrm{ad}^*$ is the coadjoint operator associated with the pairing used in
\eqref{eq:inertia_operator} (equivalently, \eqref{eq:euler_arnold} is the reduced geodesic equation
obtained from the Euler--Poincar\'e variational principle) \cite{HolmMarsdenRatiu1998}.
In matrix form, under the standard identifications on $\mathfrak{su}(N)$,
\eqref{eq:euler_arnold} can be written schematically as
\begin{equation}
\label{eq:euler_arnold_commutator}
\frac{d}{ds}\big(\mathcal{I}_{\Omegahat} \hat{\mathbf{A}}(s)\big)
=
\big[\mathcal{I}_{\Omegahat} \hat{\mathbf{A}}(s),\ \hat{\mathbf{A}}(s)\big],
\end{equation}
which reduces to $\dot{\hat{\mathbf{A}}}(s)=0$ (hence one--parameter subgroups) precisely in the
bi--invariant case $\mathcal{I}_{\Omegahat}\propto \mathrm{Id}$ \cite{Milnor1976}.

Equations \eqref{eq:euler_arnold}--\eqref{eq:euler_arnold_commutator} clarify how anisotropic penalty
tensors $\Omegahat$ lead to nontrivial optimal synthesis trajectories, as in Nielsen's geometry of
quantum computation \cite{Niel,dorth2}.

\section{Complexity Functionals for Channels: Definitions and Invariances}
\label{sec:defs-channel-noise}

A central structural issue in extending geometric complexity from unitary dynamics to open-system
dynamics is the non-uniqueness of dilations: the same reduced channel may arise from many distinct
triples consisting of an environment, an initial environment state, and a joint unitary evolution.
In this section we therefore distinguish two layers:

\begin{itemize}
\item \emph{Implementation-dependent} (or \emph{dilation-dependent}) notions, where the microscopic
      dilation is regarded as part of the data and the complexity is assigned to a particular
      physical implementation of the channel.
\item \emph{Intrinsic} notions, where one minimizes over a physically motivated class of admissible
      dilations to obtain a channel-level quantity.
\end{itemize}

\subsection{Implementation-Dependent Channel Complexity}
\label{subsec:impl-dependent}

Fix a finite-dimensional system Hilbert space $\mathscr{H}_S$ and consider a one-parameter family of
quantum channels $(\Lambda_t)_{t\ge 0}$ acting on $\mathcal{S}(\mathscr{H}_S)$.
A Stinespring dilation of $(\Lambda_t)$ is specified by the data
\begin{equation}
\label{eq:dilation-data}
\mathfrak{D}
:=
\big(\mathscr{H}_E,\ \boldsymbol{\hat{\rho}}_E,\ \boldsymbol{\hat{H}}_{tot}(s)\big),
\qquad
\boldsymbol{\hat{H}}_{tot}(s)=\boldsymbol{\hat{H}}_{tot}^{\dagger}(s)\ \text{on}\ \mathscr{H}_S\otimes \mathscr{H}_E,
\end{equation}
and yields the reduced dynamics
\begin{equation}
\label{eq:channel-from-dilation}
\Lambda_t^{(\mathfrak{D})}(\boldsymbol{\hat{\rho}}_S)
:=
\Tr_E\!\left[
\hat{\mathscr{T}}\exp\!\left(-i\int_{0}^{t}\boldsymbol{\hat{H}}_{tot}(s)\,ds\right)
(\boldsymbol{\hat{\rho}}_S\otimes \boldsymbol{\hat{\rho}}_E)\,
\hat{\mathscr{T}}\exp\!\left(i\int_{0}^{t}\boldsymbol{\hat{H}}_{tot}(s)\,ds\right)
\right].
\end{equation}
Whenever $\Lambda_t^{(\mathfrak{D})}=\Lambda_t$ for all $t$ under consideration, we say that
$\mathfrak{D}$ is a dilation of $\Lambda_t$.

\medskip

\paragraph{Embedding convention:}
When comparing system and total generators, we embed system operators as
\begin{equation}
\label{eq:embed-system}
\boldsymbol{\hat{H}}_S \equiv \boldsymbol{\hat{H}}_S \otimes \boldsymbol{\hat{I}}_E
\quad\text{on}\quad
\mathscr{H}_S\otimes \mathscr{H}_E.
\end{equation}

\medskip

\paragraph{Implementation-dependent channel complexity:}
Let $d_{tot}:=\dim(\mathscr{H}_S\otimes\mathscr{H}_E)$, and work throughout with the
$\Omegahat$-weighted non-homogeneous geometry afforded by the inner product
$\langle\cdot,\cdot\rangle_{\Omegahat}$ defined in \eqref{eq:Omega_inner_product}.
Given a dilation $\mathfrak{D}=(\mathscr{H}_E,\boldsymbol{\hat{\rho}}_E,\boldsymbol{\hat{H}}_{tot}(s))$
of $\Lambda_t$, we make the following definitions. Without loss of generality, all Hamiltonians
are taken to be traceless.

\begin{definition}[Coherent distortion complexity]
\label{eq:coherent-distortion-complexity}
Let $\Lambda_t$ be a channel trajectory on a system $S$, implemented through a Stinespring dilation
with right-trivialized generator
\begin{equation}
\boldsymbol{\hat{H}}_{\mathrm{tot}}(t)
:=
i\,\boldsymbol{\dot{\hat{U}}}_{SE}(t)\boldsymbol{\hat{U}}_{SE}(t)^\dagger.
\end{equation}
The \emph{coherent distortion complexity} of the implementation $\mathfrak{D}$ is defined as
\begin{equation}
\mathcal{G}^{\mathrm{coh}}_{\Omegahat}(\Lambda_t;\mathfrak{D})
:=
\frac{1}{\sqrt{\operatorname{Tr}(\Omegahat)}}
\int_0^t
\left\|
\boldsymbol{\hat{P}}_S\!\big(\boldsymbol{\hat{H}}_{\mathrm{tot}}(s)\big)
\right\|_{\Omegahat}\,ds,
\end{equation}
where
\begin{equation}
\boldsymbol{\hat{P}}_S(\boldsymbol{\hat{X}})
=
\frac{1}{d_E}\,\mathrm{Tr}_E(\boldsymbol{\hat{X}})\otimes \boldsymbol{\hat{I}}_E
\end{equation}
denotes the orthogonal projector onto the system-local subspace
$\mathfrak{su}(d_S)\otimes \boldsymbol{\hat{I}}_E \subset \mathfrak{su}(d_Sd_E)$.
\end{definition}

\begin{definition}[Implementation-dependent channel complexity]
\label{eq:implementation-dependent-channel-complexity}
The \emph{implementation-dependent channel complexity} of the same dilation is
\begin{equation}
\label{eq:impl-channel-complexity}
\mathcal{G}^{\mathrm{imp}}_{\Omegahat}(\Lambda_t;\mathfrak{D})
:=
\frac{1}{\sqrt{\operatorname{Tr}(\Omegahat)}}
\int_0^t
\left\|
\boldsymbol{\hat{H}}_{\mathrm{tot}}(s)
-
\boldsymbol{\hat{P}}_E\!\big(\boldsymbol{\hat{H}}_{\mathrm{tot}}(s)\big)
\right\|_{\Omegahat}\,ds,
\end{equation}
where
\begin{equation}
\boldsymbol{\hat{P}}_E(\boldsymbol{\hat{X}})
=
\frac{1}{d_S}\,\boldsymbol{\hat{I}}_S\otimes \mathrm{Tr}_S(\boldsymbol{\hat{X}})
\end{equation}
is the orthogonal projector onto the purely environmental subspace
$\boldsymbol{\hat{I}}_S\otimes \mathfrak{su}(d_E) \subset \mathfrak{su}(d_Sd_E)$.
\end{definition}

The projector $\boldsymbol{\hat{P}}_S$ extracts the component of the dilation generator that is
visible as coherent motion on the system degrees of freedom. Accordingly, the coherent distortion
complexity $\mathcal{G}^{\mathrm{coh}}_{\Omegahat}$ should not be interpreted as a universal measure of
channel complexity, but rather as a measure of \emph{coherent geometric distortion}: it quantifies
the extent to which the environment perturbs the circuit within the subsystem's coherent control
geometry. Purely decohering channels such as dephasing can be generated by interaction Hamiltonians
whose projection onto the system-local subspace vanishes, so their coherent distortion complexity
is zero even though the channel is physically nontrivial. This is illustrated explicitly in
Appendix~\ref{app:dephasing-zero-coherent-distortion}.

The implementation-dependent channel complexity $\mathcal{G}^{\mathrm{imp}}_{\Omegahat}$ provides a broader
notion: it removes only the purely environmental term $\boldsymbol{\hat{P}}_E(\boldsymbol{\hat{H}}_{\mathrm{tot}})$
and therefore remains sensitive to system-environment interaction terms responsible for noise.

\begin{remark}[Heuristic interpretation of the subtractive term]
\label{rem:heuristic-subtraction}
Intuitively, $\boldsymbol{\hat{H}}_{\mathrm{tot}}(t) - \boldsymbol{\hat{P}}_E\!\bigl(\boldsymbol{\hat{H}}_{\mathrm{tot}}(t)\bigr)$
is the part of the total generator that \emph{cannot} be attributed to the autonomous
self-evolution of the environment alone. In the standard decomposition
$\boldsymbol{\hat{H}}_{\mathrm{tot}} = \boldsymbol{\hat{H}}_S \otimes \boldsymbol{\hat{I}}_E
+ \boldsymbol{\hat{H}}_I + \boldsymbol{\hat{I}}_S \otimes \boldsymbol{\hat{H}}_E$,
the projector $\boldsymbol{\hat{P}}_E$ removes precisely the last term, retaining
the system Hamiltonian and the system-environment coupling. The implementation-dependent
complexity therefore measures the geometric cost of \emph{driving the system}: it charges
contributions from the system's own Hamiltonian and from the interaction responsible for noise,
but ignores the free rotation of the environment, which leaves the reduced channel unchanged
and should not count as a computational resource. In the closed-system limit
($\boldsymbol{\hat{H}}_I = 0$), only the system block $\boldsymbol{\hat{H}}_S\otimes\boldsymbol{\hat{I}}_E$
survives the subtraction, and the prescription recovers exactly Nielsen's unitary complexity
of $\boldsymbol{\hat{U}}_S(t)$, as required by Postulate~\textup{(P1)}.
\end{remark}

For time-independent dilations, \eqref{eq:impl-channel-complexity} simplifies to
\begin{equation}
\mathcal{G}^{\mathrm{imp}}_{\Omegahat}(\Lambda_t;\mathfrak{D})
=
\frac{t}{\sqrt{\operatorname{Tr}(\Omegahat)}}
\left\|\boldsymbol{\hat{H}}_{\mathrm{tot}}
-
\boldsymbol{\hat{P}}_E\!\big(\boldsymbol{\hat{H}}_{\mathrm{tot}}\big)\right\|_{\Omegahat}.
\end{equation}

\subsection{Restrictions on $\Omegahat$ and Invariance of $\mathcal{G}^{\mathrm{imp}}_{\Omegahat}(\Lambda_t;\mathfrak{D})$}

In order to construct a complexity measure for quantum channels that reflects physically meaningful
implementation costs, we impose a specific structure on the penalty matrix $\Omegahat$. Working in a
tensor-product Pauli basis
$\{\boldsymbol{\hat{\sigma}}_{\alpha}^{(S)}\otimes\boldsymbol{\hat{\sigma}}_{\beta}^{(E)}\}$,
where $\Omegahat_{\alpha\beta}$ denotes the $(\alpha,\beta)$ entry of the matrix $\Omegahat$,
we assume that the penalty assigned to operators carrying a non-trivial environmental component
depends only on the system index:
\begin{equation}
\Omegahat_{\alpha\beta}=\omega_{\alpha},\qquad \beta\neq 0.
\end{equation}
In other words, all operators of the form
$\boldsymbol{\hat{\sigma}}_{\alpha}^{(S)}\otimes\boldsymbol{\hat{\sigma}}_{\beta}^{(E)}$ receive the
same penalty for fixed $\alpha$, independent of the environmental label $\beta$.

\medskip
\noindent\textbf{Physical motivation:}
Quantum channels describe the effective dynamics of a system interacting with uncontrolled
environmental degrees of freedom. In practice, the environment is not directly accessible and its
operator basis carries no operational meaning for the experimenter. Different Stinespring dilations
related by a unitary transformation acting solely on the environment,
$\boldsymbol{\hat{U}}\mapsto(\boldsymbol{\hat{I}}_{S}\otimes\boldsymbol{\hat{V}}_{E})\,\boldsymbol{\hat{U}}$,
generate exactly the same reduced system dynamics. Such transformations represent a gauge freedom
in the dilation. A complexity measure should therefore be insensitive to this gauge, and the
condition $\Omegahat_{\alpha\beta}=\omega_{\alpha}$ for $\beta\neq 0$ achieves this by assigning equal
cost to all couplings that involve the same system operator and differ only in their environmental
component.

\medskip
\noindent\textbf{Geometric motivation:}
When the metric assigns equal weights to all environmental directions within each interaction block
$\boldsymbol{\hat{\sigma}}_{\alpha}^{(S)}\otimes\mathfrak{su}(d_{E})$, the adjoint action of
$\boldsymbol{\hat{I}}_S\otimes SU(d_E)$ acts by isometries of the geometry. Consequently, the
right-invariant metric on $SU(d_S d_E)$ induces a well-defined quotient Riemannian metric on
\begin{equation}
SU(d_S d_E)\big/\big(\boldsymbol{\hat{I}}_S\otimes SU(d_E)\big),
\end{equation}
in the standard sense of quotient metrics for isometric group actions \cite{lee}.

\medskip
\begin{lemma}[Block decomposition and environmental invariance]
\label{lem:block-decomposition-env-invariance}
Let $\mathfrak{g}=\mathfrak{su}(d_Sd_E)$ with orthogonal decomposition
\begin{equation}
\label{eq:lie-algebra-block-decomposition}
\mathfrak{g}
=
\bigl(\mathfrak{su}(d_S)\otimes \boldsymbol{\hat{I}}_E\bigr)
\oplus
\mathfrak{m}_{\mathrm{int}}
\oplus
\bigl(\boldsymbol{\hat{I}}_S\otimes \mathfrak{su}(d_E)\bigr),
\end{equation}
and corresponding orthogonal projectors $\boldsymbol{\hat{P}}_{S}, \boldsymbol{\hat{P}}_{int},
\boldsymbol{\hat{P}}_{E}$. Let $\|\cdot\|_{\Omegahat}$ be the norm induced by a diagonal penalty
matrix $\Omegahat$ satisfying
\begin{equation}
\label{eq:omega-isotropy-condition}
\Omegahat_{\alpha\beta}:=\|\boldsymbol{\hat{\sigma}}_{\alpha}^{(S)}\otimes \boldsymbol{\hat{\sigma}}_{\beta}^{(E)}\|^{2}_{\Omegahat}=\omega_{\alpha}
\qquad
\text{for all }\beta\neq 0.
\end{equation}
Then, for every $\boldsymbol{\hat{H}}\in\mathfrak{g}$,
\begin{equation}
\label{eq:correct-squared-norm-decomposition}
\|\boldsymbol{\hat{H}}\|_{\Omegahat}^{2}
=
\|\boldsymbol{\hat{P}}_{S}(\boldsymbol{\hat{H}})\|_{\Omegahat_{S}}^{2}
+
\|\boldsymbol{\hat{P}}_{int}(\boldsymbol{\hat{H}})\|_{\Omegahat_{int}}^{2}
+
\|\boldsymbol{\hat{P}}_{E}(\boldsymbol{\hat{H}})\|_{\Omegahat_{E}}^{2},
\end{equation}
where $\Omegahat=\Omegahat_{S}\oplus \Omegahat_{int}\oplus \Omegahat_{E}$, and moreover
\begin{equation}
\label{eq:interaction-env-isometry}
\left\|
(\boldsymbol{\hat{I}}_S\otimes \boldsymbol{\hat{V}}_E)\,
\bigl(\boldsymbol{\hat{P}}_{int}(\boldsymbol{\hat{H}})
+
\boldsymbol{\hat{P}}_{E}(\boldsymbol{\hat{H}})\bigr)\,
(\boldsymbol{\hat{I}}_S\otimes \boldsymbol{\hat{V}}_E^\dagger)
\right\|_{\Omegahat}^{2}
=
\|\boldsymbol{\hat{P}}_{int}(\boldsymbol{\hat{H}})\|_{\Omegahat_{int}}^{2}
+
\|\boldsymbol{\hat{P}}_{E}(\boldsymbol{\hat{H}})\|_{\Omegahat_{E}}^{2}
\end{equation}
for every $\boldsymbol{\hat{V}}_E\in SU(d_E)$. Consequently, the non-environmental norm
\begin{equation}
\label{eq:non-env-squared-norm}
\left\|
\boldsymbol{\hat{H}}-\boldsymbol{\hat{P}}_{E}(\boldsymbol{\hat{H}})
\right\|_{\Omegahat}^{2}
=
\|\boldsymbol{\hat{P}}_{S}(\boldsymbol{\hat{H}})\|_{\Omegahat_{S}}^{2}
+
\|\boldsymbol{\hat{P}}_{int}(\boldsymbol{\hat{H}})\|_{\Omegahat_{int}}^{2}
\end{equation}
is invariant under environmental changes of basis
$\boldsymbol{\hat{H}} \mapsto (\boldsymbol{\hat{I}}_S\otimes \boldsymbol{\hat{V}}_E)\,
\boldsymbol{\hat{H}}\,(\boldsymbol{\hat{I}}_S\otimes \boldsymbol{\hat{V}}_E^\dagger)$.
Hence the implementation-dependent channel complexity built from the integrand
$\left\|\boldsymbol{\hat{H}}_{\mathrm{tot}}(t)-\boldsymbol{\hat{P}}_{E}(\boldsymbol{\hat{H}}_{\mathrm{tot}}(t))\right\|_{\Omegahat}$
is invariant under environmental changes of basis.
\end{lemma}

\begin{proof}
Expand $\boldsymbol{\hat{H}}$ (traceless without loss of generality) in the tensor-product Pauli
basis:
\begin{equation}
\boldsymbol{\hat{H}}
=
\sum_{\alpha\neq 0} h_{\alpha 0}\,\boldsymbol{\hat{\sigma}}_{\alpha}^{(S)}\otimes \boldsymbol{\hat{I}}_{E}
+
\sum_{\alpha\neq 0,\;\beta\neq 0} h_{\alpha\beta}\,\boldsymbol{\hat{\sigma}}_{\alpha}^{(S)}\otimes \boldsymbol{\hat{\sigma}}_{\beta}^{(E)}
+
\sum_{\beta\neq 0} h_{0\beta}\,\boldsymbol{\hat{I}}_{S}\otimes \boldsymbol{\hat{\sigma}}_{\beta}^{(E)}.
\end{equation}
These three sums are precisely the system, interaction, and purely environmental components:
\begin{align}
\boldsymbol{\hat{P}}_{S}(\boldsymbol{\hat{H}})
&=
\sum_{\alpha\neq 0} h_{\alpha 0}\,\boldsymbol{\hat{\sigma}}_{\alpha}^{(S)}\otimes \boldsymbol{\hat{I}}_{E},
\\
\boldsymbol{\hat{P}}_{int}(\boldsymbol{\hat{H}})
&=
\sum_{\alpha\neq 0,\;\beta\neq 0} h_{\alpha\beta}\,\boldsymbol{\hat{\sigma}}_{\alpha}^{(S)}\otimes \boldsymbol{\hat{\sigma}}_{\beta}^{(E)},
\\
\boldsymbol{\hat{P}}_{E}(\boldsymbol{\hat{H}})
&=
\sum_{\beta\neq 0} h_{0\beta}\,\boldsymbol{\hat{I}}_{S}\otimes \boldsymbol{\hat{\sigma}}_{\beta}^{(E)}.
\end{align}
Since the metric is diagonal in this basis, \eqref{eq:correct-squared-norm-decomposition} follows
immediately. Now let $\boldsymbol{\hat{W}}:=\boldsymbol{\hat{I}}_{S}\otimes \boldsymbol{\hat{V}}_{E}$.
Conjugation by $\boldsymbol{\hat{W}}$ leaves the system block invariant and rotates the environmental
basis inside each interaction block and the purely environmental sector via an orthogonal matrix
$\boldsymbol{\hat{R}}$ (the adjoint representation of $SU(d_E)$ on $\mathfrak{su}(d_E)$):
\begin{equation}
h_{\alpha\beta}\mapsto \sum_{\gamma\neq 0}R_{\beta\gamma}h_{\alpha\gamma}, \qquad
h_{0\beta}\mapsto \sum_{\gamma\neq 0}R_{\beta\gamma}h_{0\gamma}.
\end{equation}
Using \eqref{eq:omega-isotropy-condition}, for each fixed $\alpha\neq 0$:
\begin{equation}
\sum_{\beta\neq 0}\Omegahat_{\alpha\beta}|h_{\alpha\beta}|^{2}
=
\omega_{\alpha}\sum_{\beta\neq 0}|h_{\alpha\beta}|^{2},
\end{equation}
and since $\boldsymbol{\hat{R}}$ is orthogonal, Euclidean lengths are preserved. Therefore the
interaction and environmental contributions to $\|\boldsymbol{\hat{H}}\|^2_{\Omegahat}$ are invariant
under $\boldsymbol{\hat{V}}_E$. This proves \eqref{eq:interaction-env-isometry} and
\eqref{eq:non-env-squared-norm}, and hence the invariance of
$\mathcal{G}^{\mathrm{imp}}_{\Omegahat}(\Lambda_t;\mathfrak{D})$.
\end{proof}

\paragraph{Minimal invariance: environmental change of basi:}
The Kraus operators of $\Lambda_t$ depend on the chosen orthonormal basis of $\mathscr{H}_E$, yet the
channel itself does not. The next Corollary shows that the implementation-dependent complexity
Eq.(\ref{eq:impl-channel-complexity}) is invariant under those changes of basis that leave the channel
unchanged.

\begin{corollary}[Invariance of $\mathcal{G}^{\mathrm{imp}}_{\Omegahat}$ under environment changes of basis]
\label{lem:env-basis-invariance}
Let $\mathfrak{D}=(\mathscr{H}_E,\boldsymbol{\hat{\rho}}_E,\boldsymbol{\hat{H}}_{tot}(s))$ be a
dilation of $\Lambda_t$ and let $\boldsymbol{\hat{V}}_E$ be any unitary on $\mathscr{H}_E$.
Define the transformed dilation
\begin{equation}
\label{eq:env-basis-transformed-dilation}
\mathfrak{D}'
:=
\Big(\mathscr{H}_E,\ \boldsymbol{\hat{\rho}}'_E,\ \boldsymbol{\hat{H}}'_{tot}(s)\Big),
\quad
\boldsymbol{\hat{\rho}}'_E:=\boldsymbol{\hat{V}}_E\boldsymbol{\hat{\rho}}_E \boldsymbol{\hat{V}}_E^{\dagger},
\quad
\boldsymbol{\hat{H}}'_{tot}:=(\boldsymbol{\hat{I}}_S\otimes \boldsymbol{\hat{V}}_E)\boldsymbol{\hat{H}}_{tot}(\boldsymbol{\hat{I}}_S\otimes \boldsymbol{\hat{V}}_E^{\dagger}).
\end{equation}
Then $\Lambda_t^{(\mathfrak{D}')}=\Lambda_t^{(\mathfrak{D})}$ for all $t$, and
\begin{equation}
\label{eq:impl-invariance-statement}
\mathcal{G}^{\mathrm{imp}}_{\Omegahat}\!\big(\Lambda_t;\mathfrak{D}'\big)
=
\mathcal{G}^{\mathrm{imp}}_{\Omegahat}\!\big(\Lambda_t;\mathfrak{D}\big)
\qquad\text{for all }t.
\end{equation}
\end{corollary}

\begin{proof}
Let $\boldsymbol{\hat{W}}:= \boldsymbol{\hat{I}}_S \otimes \boldsymbol{\hat{V}}_E$ and
$\boldsymbol{\hat{U}}'(t) := \boldsymbol{\hat{W}}\boldsymbol{\hat{U}}(t)\boldsymbol{\hat{W}}^\dagger$.
Channel equality follows from cyclicity of the partial trace. For the generator, differentiating
$\boldsymbol{\hat{U}}'(t)$ gives
$\boldsymbol{\hat{H}}'_{tot}(t) = \boldsymbol{\hat{W}}\boldsymbol{\hat{H}}_{tot}(t)\boldsymbol{\hat{W}}^\dagger$,
and hence
\begin{equation}
\boldsymbol{\hat{H}}'_{tot}(t)-\boldsymbol{\hat{P}}_{E}(\boldsymbol{\hat{H}}'_{tot}(t))
=\boldsymbol{\hat{W}}
\bigl(\boldsymbol{\hat{H}}_{tot}(t)-\boldsymbol{\hat{P}}_{E}(\boldsymbol{\hat{H}}_{tot}(t))\bigr)
\boldsymbol{\hat{W}}^\dagger.
\end{equation}
By direct application of Lemma~\ref{lem:block-decomposition-env-invariance},
$\|\boldsymbol{\hat{H}}'_{tot}(t)-\boldsymbol{\hat{P}}_{E}(\boldsymbol{\hat{H}}'_{tot}(t))\|_{\Omegahat}
= \|\boldsymbol{\hat{H}}_{tot}(t)-\boldsymbol{\hat{P}}_{E}(\boldsymbol{\hat{H}}_{tot}(t))\|_{\Omegahat}$,
which proves \eqref{eq:impl-invariance-statement}.
\end{proof}

\subsection{Noise Complexity}
\label{sec:noisecomplexity}

Let us consider the scenario where an ideal closed-system evolution generated by
$\boldsymbol{\hat{H}}_S(t)$ is implemented in the presence of noise. The ideal evolution is
\begin{equation}
\boldsymbol{\hat{U}}_S(t)
=
\hat{\mathscr{T}}
\exp\!\left(
-i\int_0^t \boldsymbol{\hat{H}}_S(s)\,ds
\right),
\end{equation}
embedded on the total Hilbert space as $\boldsymbol{\hat{U}}_{S}(t)\otimes \boldsymbol{\hat{I}}_{E}$.
The noisy channel $\Lambda_t^{q}$ ($q\geq 0$ is the noise strength, with
$\Lambda_t^{0}(\cdot)=\boldsymbol{\hat{U}}_S(t)(\cdot)\boldsymbol{\hat{U}}_S^\dagger(t)$) is
realized by a dilation $\mathfrak{D}$ with unitary $\boldsymbol{\hat{U}}_q(t)$ and generator
\begin{equation}
\boldsymbol{\hat{H}}_{tot,q}(t)
=
i\,(\partial_{t} \boldsymbol{\hat{U}}_{q}(t))\boldsymbol{\hat{U}}_{q}(t)^\dagger.
\end{equation}
Subtracting the purely environmental projection yields the effective generator
$\boldsymbol{\hat{H}}_{tot,q}(s)-\boldsymbol{\hat{P}}_{E}(\boldsymbol{\hat{H}}_{tot,q}(s))$,
which induces the projected trajectory
\begin{equation}
\label{eqn:noiseunitary}
\boldsymbol{\hat{U}}_{proj,q}(t)
=
\hat{\mathscr{T}}
\exp\!\left(
-i\int_0^t \bigl[\boldsymbol{\hat{H}}_{tot,q}(s)-\boldsymbol{\hat{P}}_{E}(\boldsymbol{\hat{H}}_{tot,q}(s))\bigr]\,ds
\right)
\end{equation}
in $SU(d_{S}d_{E})\big/\big(\boldsymbol{\hat{I}}_{S}\otimes SU(d_{E})\big)$.

\begin{definition}[Geometric noise complexity]
\label{eqn:noisecomplexity}
The \emph{geometric noise complexity} of the implementation $\mathfrak{D}$ is
\begin{equation}
\mathcal{N}_{\Omegahat}(\Lambda_t;\mathfrak{D})
=
\mathscr{D}_{\Omegahat}\!\left(
\boldsymbol{\hat{U}}_S(t),\,
\boldsymbol{\hat{U}}_{proj,q}(t)
\right),
\end{equation}
where $\mathscr{D}_{\Omegahat}$ is the Carnot--Carath\'eodory distance induced by the metric
$\langle\cdot,\cdot\rangle_{\Omegahat}$.
\end{definition}

This quantity measures the minimal control cost required to move from the implemented system
trajectory back to the ideal one within the manifold
$SU(d_Sd_E)\big/\big(\boldsymbol{\hat{I}}_S\otimes SU(d_E)\big)$.

\begin{proposition}[Hilbert-Schmidt lower bound for noise complexity]
\label{prop:boundhsomega}
Let $\mathscr{D}_{HS}$ denote the geodesic distance on $SU(d_S d_E)$ induced by the
Hilbert-Schmidt metric. Then
\begin{equation}
\mathscr{D}_{HS}\!\left(
\boldsymbol{\hat{U}}_S(t),\,
\boldsymbol{\hat{U}}_{proj,q}(t)
\right)
\le
\mathcal{N}_{\Omegahat}(\Lambda_t;\mathfrak{D}).
\end{equation}
\end{proposition}

\begin{proof}
Since the penalty operator $\Omegahat$ has weights $l_i\ge 1$, the norms satisfy
$\|\hat{\mathbf{H}}\|_{HS} \le \sqrt{\langle\hat{\mathbf{H}},\hat{\mathbf{H}}\rangle_{\Omegahat}}$
for all $\hat{\mathbf{H}}\in\mathfrak{su}(N)$. Therefore for any smooth curve $\gamma$
connecting $\boldsymbol{\hat{U}}_S(t)$ and $\boldsymbol{\hat{U}}_{proj,q}(t)$,
$L_{HS}(\gamma)\le L_{\Omegahat}(\gamma)$. Taking the infimum over all such curves yields the result.
\end{proof}

\paragraph{Total implementation complexity:}
A measure of the total control resources required to implement a computation and mitigate noise is
\begin{equation}
\mathcal{C}(\Lambda_{t}^{q};\mathfrak{D})
:=
\mathscr{G}_{\Omegahat}(\boldsymbol{\hat{U}}_S(t))
+
\mathcal{N}_{\Omegahat}(\Lambda^{q}_t;\mathfrak{D}),
\end{equation}
where the first term quantifies the cost of the ideal algorithm and the second quantifies the
additional resources needed to mitigate noise.

\subsection{Intrinsic channel complexity}
\label{subsec:intrinsic}

The quantity $\mathcal{G}^{\mathrm{imp}}_{\Omegahat}(\Lambda_t;\mathfrak{D})$ assigns a complexity
value to a \emph{particular} dilation $\mathfrak{D}$ realizing the channel trajectory
$(\Lambda_t)_{t\in[0,T]}$. Although it is invariant under environmental unitary changes of basis,
a given channel trajectory may still admit many physically inequivalent realizations. It is therefore
natural to define a channel-level notion of complexity by minimizing over a prescribed class of
physically admissible dilations.

The key point is that such an intrinsic quantity is necessarily valued relative to a chosen
\emph{resource model}. Different experimental restrictions may lead to different classes of
admissible realizations, and hence to different minimal costs.

\subsubsection*{Admissible dilations and resource models}

Fix a time window $[0,T]$ with $T>0$, and a family of channels $(\Lambda_t)_{t\in[0,T]}$ on
$\mathscr{H}_S$. A \emph{resource model} $\mathcal{R}$ specifies which dilation data are regarded
as physically available. This may include bounds on environmental dimension, Hamiltonian strength,
preparation energy, locality structure, interaction rank, coupling graph, or time dependence.

Given such a resource model, an \emph{admissible dilation} is a triple
\begin{equation}
\label{eq:adm_dilation_data}
\mathfrak{D}
=
\bigl(\mathscr{H}_E,\boldsymbol{\hat{\rho}}_E,\boldsymbol{\hat{H}}_{\mathrm{tot}}(t)\bigr)
\end{equation}
such that $\Lambda_t^{(\mathfrak{D})}=\Lambda_t$ for all $t\in[0,T]$, and the dilation data
satisfy the constraints of $\mathcal{R}$. A concrete choice is
\begin{equation}
\label{eq:resource_constraints}
\dim(\mathscr{H}_E)\le d_E^{\max},
\qquad
\|\boldsymbol{\hat{H}}_{\mathrm{tot}}(t)\|_{\mathrm{op}}\le J_{\max}
\quad\text{for all } t\in[0,T].
\end{equation}
If an explicit environmental Hamiltonian $\boldsymbol{\hat{H}}_E$ is part of the model, one may
further constrain the initial preparation energy:
\begin{equation}
\label{eq:env_energy_constraint}
\Tr(\boldsymbol{\hat{\rho}}_E\,\boldsymbol{\hat{H}}_E(t))\le E_{\max}.
\end{equation}

\begin{definition}[Admissible dilation set]
\label{def:admissible-dilation-set}
Given $(\Lambda_t)_{t\in[0,T]}$ and a resource model $\mathcal{R}$, denote by
$\mathfrak{Dil}_{\mathcal{R}}^{[0,T]}(\Lambda_{t})$ the set of all dilations satisfying:
\begin{enumerate}
\item \textbf{Exact realization on $[0,T]$:} $\Lambda_t^{(\mathfrak{D})}=\Lambda_t$ for every
      $t\in[0,T]$.
\item \textbf{Resource admissibility:} the dilation data satisfy the constraints of $\mathcal{R}$.
\end{enumerate}
If the admissible set is empty, set the intrinsic complexity equal to $+\infty$.
\end{definition}

\begin{definition}[Intrinsic channel complexity relative to $\mathcal{R}$]
\label{def:intrinsic-channel-complexity}
The \emph{intrinsic channel complexity} of $(\Lambda_t)_{t\in[0,T]}$, relative to $\mathcal{R}$, is
\begin{equation}
\label{eq:intrinsic-channel-complexity}
\mathcal{G}_{\mathrm{intr}}^{\mathcal{R}}(\Lambda_t;[0,T])
:=
\inf_{\mathfrak{D}\in \mathfrak{Dil}_{\mathcal{R}}^{[0,T]}(\Lambda)}
\mathcal{G}^{\mathrm{imp}}_{\Omegahat}\!\bigl(\Lambda_t;\mathfrak{D}\bigr).
\end{equation}
When the time window is clear, we write $\mathcal{G}_{\mathrm{intr}}^{\mathcal{R}}(\Lambda_t)$.
\end{definition}

This definition makes explicit that intrinsic channel complexity is \emph{resource-relative}.
If $\mathcal{R}_1 \subseteq \mathcal{R}_2$, then
$\mathcal{G}_{\mathrm{intr}}^{\mathcal{R}_2}(\Lambda_t;[0,T])
\le \mathcal{G}_{\mathrm{intr}}^{\mathcal{R}_1}(\Lambda_t;[0,T])$:
enlarging the admissible class can only decrease the intrinsic complexity.

\subsubsection*{Attainment of the infimum}

\begin{proposition}[Existence of an optimal dilation]
\label{prop:optimal_dilation}
Assume the admissible set $\mathfrak{Dil}_{\mathcal{R}}^{[0,T]}(\Lambda_t)$ is defined by
\eqref{eq:resource_constraints}, is nonempty, and that the map
$(\boldsymbol{\hat{\rho}}_E,\boldsymbol{\hat{H}}_{\mathrm{tot}})\mapsto\Lambda_t^{(\mathfrak{D})}$
is continuous and the functional
$\mathfrak{D}\mapsto\mathcal{G}^{\mathrm{imp}}_{\Omegahat}(\Lambda_t;\mathfrak{D})$ is lower
semicontinuous on $\mathfrak{Dil}_{\mathcal{R}}^{[0,T]}(\Lambda_t)$. Then there exists an optimal
admissible dilation $\mathfrak{D}^*$ such that
$\mathcal{G}^{\mathrm{imp}}_{\Omegahat}(\Lambda_t;\mathfrak{D}^*)=\mathcal{G}_{\mathrm{intr}}^{\mathcal{R}}(\Lambda_t;[0,T])$.
\end{proposition}

\begin{proof}
The constraints \eqref{eq:resource_constraints} embed all admissible dilations into a compact
parameter space (finite-dimensional, bounded). The continuity assumption implies the exact-realization
subset is closed, hence compact. The complexity functional is lower semicontinuous by assumption, with
the map $\boldsymbol{\hat{H}}_{\mathrm{tot}}(t)\mapsto\|\boldsymbol{\hat{H}}_{\mathrm{tot}}(t)-\boldsymbol{\hat{P}}_E(\boldsymbol{\hat{H}}_{\mathrm{tot}}(t))\|_{\Omegahat}$
being continuous and convex for each fixed $t$. The standard direct method yields the minimizer.
\end{proof}

\begin{remark}[Resource structure and admissible implementations]
\label{rem:resource_structure}
The intrinsic channel complexity is resource-relative in two senses: relative to the admissible
class $\mathcal{R}$ (which specifies physically available implementations) and relative to the
geometric weighting $\Omegahat$ (which assigns costs to operator directions). Different choices of
$\mathcal{R}$ or $\Omegahat$ may yield different intrinsic complexities for the same channel.
\end{remark}

\begin{remark}[GKSL-motivated admissible classes]
\label{rem:GKSL_admissible_class}
When $(\Lambda_t)_{t\in[0,T]}$ is a quantum dynamical semigroup generated by a GKSL equation, it
is natural to restrict $\mathcal{R}$ to dilations compatible with the weak-coupling physical model,
i.e.\ Hamiltonians of the form
$\boldsymbol{\hat{H}}_{\mathrm{tot}}=\boldsymbol{\hat{H}}_{S}+\boldsymbol{\hat{H}}_{E}+\varepsilon\,\boldsymbol{\hat{H}}_{I}$
with stationary environmental preparations and bounded coupling $\varepsilon$. In this setting,
$\mathcal{G}_{\mathrm{intr}}^{\mathcal{R}}(\Lambda_t;[0,T])$ measures the minimal implementation
cost among dilations compatible with this Markovian model. We formalize this further in
Section~\ref{sec:lindblad}.
\end{remark}

\subsubsection*{Intrinsic noise complexity}

The intrinsic counterpart of the noise complexity is
\begin{equation}
\label{eq:intrinsic-noise-complexity}
\mathcal{N}_{\mathrm{intr}}^{\mathcal{R}}(\Lambda_t;[0,T])
:=
\inf_{\mathfrak{D}\in \mathfrak{Dil}_{\mathcal{R}}^{[0,T]}(\Lambda_t)}
\mathcal{N}_{\Omegahat}\!\bigl(\Lambda_t;\mathfrak{D}\bigr),
\qquad t\in[0,T].
\end{equation}
Under the resource constraints of Proposition~\ref{prop:optimal_dilation} and nonemptiness of the
admissible set, the infimum is attained by the same compactness argument, since
$\mathfrak{D}\mapsto\mathcal{N}_{\Omegahat}(\Lambda_t;\mathfrak{D})$ is continuous.

\section{Postulates and variational characterization of channel complexity}
\label{sec:axioms-subtractive}

The purpose of this section is twofold. First, we formulate a set of postulates that any
dilation-based notion of channel complexity should satisfy. Second, we show that, within the
geometric framework adopted in this work, these postulates naturally lead to the subtractive form
\begin{equation}
\boldsymbol{\hat{H}}_{\mathrm{eff}}(t)
=
\boldsymbol{\hat{H}}_{\mathrm{tot}}(t)
-
\boldsymbol{\hat{P}}_E\!\bigl(\boldsymbol{\hat{H}}_{\mathrm{tot}}(t)\bigr),
\end{equation}
and hence to the implementation-dependent channel complexity introduced in
Definition~\ref{eq:implementation-dependent-channel-complexity}.

\subsection{Postulates for dilation-based channel complexity}
\label{subsec:axioms}

Fix a dilation implementation
$\mathfrak{D}=(\mathscr{H}_E,\boldsymbol{\hat{\rho}}_E,\boldsymbol{\hat{H}}_{\mathrm{tot}}(t))$
of a channel trajectory $(\Lambda_t)_{t\in[0,T]}$, and let
$\boldsymbol{\hat{U}}_{\mathrm{tot}}(t)$ denote the corresponding total unitary propagator. We
regard a channel complexity as a nonnegative functional
$\mathcal{G}_{\Omegahat}(\Lambda_t;\mathfrak{D})$ associated with the chosen dilation.

\begin{definition}[Postulates for a dilation-based channel complexity]
\label{def:desiderata}
A dilation-based channel complexity functional
$\mathcal{G}_{\Omegahat}(\Lambda_t;\mathfrak{D})$ is required to satisfy the following properties.

\begin{enumerate}
\item \textbf{(P1) Closed-system consistency.}
If the reduced dynamics is exactly unitary on the system,
$\Lambda_t(\boldsymbol{\hat{\rho}})=\boldsymbol{\hat{U}}_S(t)\boldsymbol{\hat{\rho}}\boldsymbol{\hat{U}}_S(t)^\dagger$,
then the channel complexity must reduce to the Nielsen-type geometric complexity of the system
unitary:
\begin{equation}
\mathcal{G}_{\Omegahat}(\Lambda_t;\mathfrak{D})
=
\mathscr{G}_{\Omegahat}\!\bigl(\boldsymbol{\hat{U}}_S(t)\bigr).
\end{equation}

\item \textbf{(P2) Environment-only neutrality.}
If the total generator is purely environmental,
$\boldsymbol{\hat{H}}_{\mathrm{tot}}(t)=\boldsymbol{\hat{I}}_S\otimes \boldsymbol{\hat{H}}_E(t)$,
so that the induced system channel is the identity, then
\begin{equation}
\mathcal{G}_{\Omegahat}(\Lambda_t;\mathfrak{D})=0.
\end{equation}

\item \textbf{(P3) Stability under environmental change of basis.}
If $\boldsymbol{\hat{V}}_E$ is a unitary on $\mathscr{H}_E$ and one transforms the dilation data
as in Corollary~\ref{lem:env-basis-invariance}, so that the reduced system dynamics is unchanged,
then the complexity must be invariant:
\begin{equation}
\mathcal{G}_{\Omegahat}(\Lambda_t;\mathfrak{D}')
=
\mathcal{G}_{\Omegahat}(\Lambda_t;\mathfrak{D}).
\end{equation}

\item \textbf{(P4) Variational reduction principle.}
At each time $t$, the physically relevant generator is obtained by minimizing the instantaneous
cost over all additions of purely environmental terms. More precisely, among all generators of the
form $\boldsymbol{\hat{H}}_{\mathrm{tot}}(t)-\boldsymbol{\hat{I}}_S\otimes \boldsymbol{\hat{K}}_E(t)$,
the effective generator $\boldsymbol{\hat{H}}_{\mathrm{eff}}(t)$ is defined by
\begin{equation}
\boldsymbol{\hat{H}}_{\mathrm{eff}}(t)
=
\operatorname*{arg\,min}_{\boldsymbol{\hat{K}}_E(t)}
\left\|
\boldsymbol{\hat{H}}_{\mathrm{tot}}(t)
-
\boldsymbol{\hat{I}}_S\otimes \boldsymbol{\hat{K}}_E(t)
\right\|_{\Omegahat}.
\end{equation}
The channel complexity of the dilation is then the length of this reduced generator:
\begin{equation}
\mathcal{G}_{\Omegahat}(\Lambda_t;\mathfrak{D})
=
\frac{1}{\sqrt{\operatorname{Tr}(\Omegahat)}}
\int_0^T
\left\|
\boldsymbol{\hat{H}}_{\mathrm{eff}}(t)
\right\|_{\Omegahat}\,dt.
\end{equation}
\end{enumerate}
\end{definition}

\subsection*{Motivation for the variational principle}

The postulates above imply that purely environmental generators should be treated as gauge
directions. Corollary~\ref{lem:env-basis-invariance} shows that both the reduced channel
trajectory and the implementation-dependent channel complexity are invariant under environmental
changes of basis, meaning the physically relevant content of a dilation is not tied to a
particular environmental frame.

\begin{remark}[Generator-level interpretation]
\label{rem:gauge-interpretation}
At the level of generators, purely environmental directions should be regarded as gauge degrees of
freedom: modifying the dilation by terms acting only on the environment changes the environmental
representation without altering the induced system channel. Thus the quantity
$\boldsymbol{\hat{H}}_{\mathrm{tot}}(t)-\boldsymbol{\hat{P}}_E(\boldsymbol{\hat{H}}_{\mathrm{tot}}(t))$
is the component of the total generator relevant for the system-level implementation cost, with the
purely environmental part removed. The variational principle in (P4) simply selects, at each time,
the minimal-cost representative of this equivalence class.

Put differently, $\boldsymbol{\hat{H}}_{\mathrm{eff}}(t)$ is the \emph{active} part of the total
generator from the system's perspective: it comprises the system Hamiltonian and the interaction
terms that produce transitions in the reduced dynamics, while the passive environmental
self-rotation --- which merely re-labels the ancilla degrees of freedom without affecting the
induced channel --- is quotiented out. The channel complexity is therefore an integrated
\emph{driving cost}: the geometric length accumulated by the components of the joint unitary that
genuinely affect the system evolution.
\end{remark}

\subsection{Variational characterization}
\label{subsec:variational-reduction}

We now show that the variational principle of Postulate (P4) leads exactly to the subtractive form
of the implementation-dependent channel complexity.

\begin{theorem}[Variational characterization of the effective generator]
\label{prop:variational-reduction}
Let $\boldsymbol{\hat{H}}_{\mathrm{tot}}(t)$ be the generator of a dilation of $\Lambda_t$, and
suppose that $\Omegahat$ is chosen so that purely environmental directions are orthogonal to their
complement (condition \eqref{eq:omega-isotropy-condition}). Consider the variational problem
\begin{equation}
\min_{\boldsymbol{\hat{K}}_E(t)}
\left\|
\boldsymbol{\hat{H}}_{\mathrm{tot}}(t)
-
\boldsymbol{\hat{I}}_S\otimes \boldsymbol{\hat{K}}_E(t)
\right\|_{\Omegahat},
\end{equation}
where the minimization is over all environmental generators
$\boldsymbol{\hat{K}}_E(t)\in\mathfrak{su}(\mathscr{H}_E)$. Then the unique minimizer is
\begin{equation}
\boldsymbol{\hat{K}}_E^*(t)
=
\frac{1}{d_S}\operatorname{Tr}_S\!\bigl(\boldsymbol{\hat{H}}_{\mathrm{tot}}(t)\bigr),
\end{equation}
and the corresponding effective generator is
\begin{equation}
\boldsymbol{\hat{H}}_{\mathrm{eff}}(t)
=
\boldsymbol{\hat{H}}_{\mathrm{tot}}(t)
-
\boldsymbol{\hat{P}}_E\!\bigl(\boldsymbol{\hat{H}}_{\mathrm{tot}}(t)\bigr).
\end{equation}
Consequently,
\begin{equation}
\mathcal{G}^{\mathrm{imp}}_{\Omegahat}(\Lambda_t;\mathfrak{D})
=
\frac{1}{\sqrt{\operatorname{Tr}(\Omegahat)}}
\int_0^T
\left\|
\boldsymbol{\hat{H}}_{\mathrm{tot}}(t)
-
\boldsymbol{\hat{P}}_E\!\bigl(\boldsymbol{\hat{H}}_{\mathrm{tot}}(t)\bigr)
\right\|_{\Omegahat}\,dt
\end{equation}
satisfies Postulates \textup{(P1)}--\textup{(P4)}.
\end{theorem}

\begin{proof}
Let $\mathcal{Q}_E=\{\boldsymbol{\hat{I}}_S\otimes \boldsymbol{\hat{K}}_E:\boldsymbol{\hat{K}}_E\in\mathfrak{su}(\mathscr{H}_E)\}$
be the subspace of purely environmental traceless generators. The variational problem is to find the
point in the affine space
$\boldsymbol{\hat{H}}_{\mathrm{tot}}(t)-\mathcal{Q}_E$ with minimal $\Omegahat$-norm.

Under condition \eqref{eq:omega-isotropy-condition}, the orthogonal projection onto
$\mathcal{Q}_E$ is precisely $\boldsymbol{\hat{P}}_E$, and every operator admits the orthogonal
decomposition
$\boldsymbol{\hat{X}} = \boldsymbol{\hat{P}}_E(\boldsymbol{\hat{X}}) + (\boldsymbol{\hat{X}}-\boldsymbol{\hat{P}}_E(\boldsymbol{\hat{X}}))$.
Therefore
$\|\boldsymbol{\hat{H}}_{\mathrm{tot}}(t)-\boldsymbol{\hat{I}}_S\otimes\boldsymbol{\hat{K}}_E(t)\|_{\Omegahat}$
is minimized exactly when
$\boldsymbol{\hat{I}}_S\otimes\boldsymbol{\hat{K}}_E(t)=\boldsymbol{\hat{P}}_E(\boldsymbol{\hat{H}}_{\mathrm{tot}}(t))$,
giving $\boldsymbol{\hat{K}}_E^*(t)=\frac{1}{d_S}\operatorname{Tr}_S(\boldsymbol{\hat{H}}_{\mathrm{tot}}(t))$.
Uniqueness follows from the uniqueness of orthogonal projection onto a closed linear subspace.

We now verify the postulates.

\emph{(P1):} If the dynamics is unitary, one may write
$\boldsymbol{\hat{H}}_{\mathrm{tot}}(t)=\boldsymbol{\hat{H}}_S(t)\otimes\boldsymbol{\hat{I}}_E+\boldsymbol{\hat{I}}_S\otimes\boldsymbol{\hat{H}}_E(t)$.
Subtracting $\boldsymbol{\hat{P}}_E$ leaves $\boldsymbol{\hat{H}}_{\mathrm{eff}}(t)=\boldsymbol{\hat{H}}_S(t)\otimes\boldsymbol{\hat{I}}_E$,
and the channel complexity reduces to the system-unitary complexity.

\emph{(P2):} If $\boldsymbol{\hat{H}}_{\mathrm{tot}}(t)=\boldsymbol{\hat{I}}_S\otimes\boldsymbol{\hat{H}}_E(t)$,
then $\boldsymbol{\hat{P}}_E(\boldsymbol{\hat{H}}_{\mathrm{tot}}(t))=\boldsymbol{\hat{H}}_{\mathrm{tot}}(t)$,
so $\boldsymbol{\hat{H}}_{\mathrm{eff}}(t)=0$ and $\mathcal{G}^{\mathrm{imp}}_{\Omegahat}(\Lambda_t;\mathfrak{D})=0$.

\emph{(P3):} Follows directly from Corollary~\ref{lem:env-basis-invariance}.

\emph{(P4):} Holds by construction, since the complexity is defined as the integral of the norm
of the variationally selected effective generator.
\end{proof}

The theorem shows that the subtractive term $\boldsymbol{\hat{P}}_E(\boldsymbol{\hat{H}}_{\mathrm{tot}}(t))$
is not an ad hoc modification, but the unique outcome of imposing environment-only neutrality
together with invariance under environmental gauge directions and the variational reduction
principle.

\begin{remark}[Structural properties inherited by intrinsic channel complexity]
\label{rem:intrinsic-postulates}
The intrinsic channel complexity
$\mathcal{G}_{\mathrm{intr}}^{\mathcal{R}}(\Lambda_t;[0,T])$
inherits, under natural assumptions on the admissible class, the analogues of
Postulates~\textup{(P1)}--\textup{(P3)}: if the admissible class contains the trivial unitary
realization whenever the reduced dynamics is unitary, then closed-system consistency holds;
if it contains purely environmental realizations of the identity channel, then environment-only
neutrality holds; and if it is invariant under environmental changes of basis, then gauge
stability holds.

At the intrinsic level, however, Postulate~\textup{(P4)} takes a different form. The
implementation-dependent functional already satisfies the variational reduction principle for a
fixed dilation; the intrinsic quantity minimizes this cost over the admissible class. The natural
fourth postulate for intrinsic complexity is therefore:

\medskip
\noindent\textbf{(P4) Resource monotonicity.}
If two resource models satisfy
$\mathfrak{Dil}_{\mathcal{R}_1}^{[0,T]}(\Lambda_t)\subseteq\mathfrak{Dil}_{\mathcal{R}_2}^{[0,T]}(\Lambda_t)$,
then
$\mathcal{G}_{\mathrm{intr}}^{\mathcal{R}_2}(\Lambda_t;[0,T])\le\mathcal{G}_{\mathrm{intr}}^{\mathcal{R}_1}(\Lambda_t;[0,T])$.

This is automatically satisfied: enlarging the admissible resource class enlarges the set over
which the infimum is taken, so the intrinsic complexity can only decrease. The physical motivation
is transparent: relaxing available resource constraints can only make the implementation problem
easier, never harder.
\end{remark}

\section{Coherence power and geometric lower bounds}
\label{sec:coherence-bounds}

This section develops an operational interpretation of the geometric complexity functionals by
relating them to the ability of unitary dynamics to generate or suppress coherence relative to a
fixed reference basis. The basic idea is that any unitary capable of creating or destroying a
prescribed amount of coherence must incur a corresponding geometric cost, yielding a
basis-dependent lower bound for unitary geometric complexity in terms of cohering/decohering power.

The same perspective is particularly useful for noise complexity. The relative noise unitary
(see \eqref{eqn:noiseunitary}) comparing the implemented projected evolution to the ideal system
evolution carries information about how strongly noise changes coherence. A coherence-based lower
bound for that relative unitary therefore yields a nontrivial lower bound for noise complexity,
complementing the deficit-type estimates derived later in Section~\ref{sec:lindblad}.

Furthermore, this coherence-based viewpoint helps motivate complexity growth with system dimension:
as the dimension increases, coherent dynamics relative to a fixed basis become richer, and a
unitary that produces or suppresses coherence across a larger state space must pay a corresponding
geometric cost.

\subsection{Coherence lower bound for unitary geometric complexity}
\label{subsec:coherence-lower-bound}

\begin{definition}[Relative entropy of coherence]
\label{def:relative-entropy-coherence}
Let $\mathscr{E}(\cdot):=\sum_{i}\boldsymbol{\hat{P}}_{i}(\cdot)\boldsymbol{\hat{P}}_{i}$ be a
completely dephasing channel, where $\{\boldsymbol{\hat{P}}_{i}\}$ is an orthogonal set of
projectors spanning $\mathscr{H}$. For $\boldsymbol{\hat{\rho}}\in\mathcal{S}(\mathscr{H})$, the
\emph{linear-entropy coherence} is
\begin{equation}
\label{def:entropycoherencemeasure}
C_{\mathscr{E}}(\boldsymbol{\hat{\rho}})
:=
S_L\!\bigl(\mathscr{E}(\boldsymbol{\hat{\rho}})\bigr)-S_L(\boldsymbol{\hat{\rho}})
=
\Tr(\boldsymbol{\hat{\rho}}^2)-\Tr\!\bigl(\mathscr{E}(\boldsymbol{\hat{\rho}})^2\bigr),
\end{equation}
where $S_L(\boldsymbol{\hat{\rho}}):=1-\Tr(\boldsymbol{\hat{\rho}}^2)$ is the linear entropy.
\end{definition}

\begin{definition}[Cohering power of a unitary]
\label{eqn:coheringpower}
The \emph{cohering power} of $\boldsymbol{\hat{U}}$ with respect to $\mathscr{E}$ at time $t\ge 0$
is
\begin{equation}
\label{eqn:decpower}
\mathscr{C}_{\mathscr{E}}(\boldsymbol{\hat{U}})
:=
\max_{\boldsymbol{\hat{\rho}}\in\mathcal{S}(\mathscr{H})}
\Big|C_{\mathscr{E}}(\boldsymbol{\hat{U}}\boldsymbol{\hat{\rho}}\boldsymbol{\hat{U}}^{\dagger})
-C_{\mathscr{E}}(\boldsymbol{\hat{\rho}})\Big|,
\end{equation}
where the maximum is taken over all density operators on $\mathcal{S}(\mathscr{H})$.
\end{definition}

The total variation of $C_{\mathscr{E}}$ under a unitary evolution is controlled from above by the
geometric length of that evolution, yielding the following lower bound.

\begin{theorem}[Coherence-power lower bound for unitary geometric complexity]
\label{thm:coherence-lower-bound}
Let $\boldsymbol{\hat{U}}(t)=e^{-it\boldsymbol{\hat{H}}}$ be a unitary evolution on a
finite-dimensional Hilbert space $\mathscr{H}$ with time-independent generator
$\boldsymbol{\hat{H}}=\boldsymbol{\hat{H}}^\dagger$, and let $\mathscr{E}$ denote the complete
dephasing channel in the fixed reference basis. Then, for every initial state
$\boldsymbol{\hat{\rho}}\in\mathcal{S}(\mathscr{H})$ and every $t\ge 0$, the unitary geometric
complexity in the Hilbert-Schmidt geometry satisfies
\begin{equation}
\label{eq:coherence-lower-bound}
\mathscr{G}_{HS}\!\bigl(\boldsymbol{\hat{U}}(t)\bigr)
\;\ge\;
\frac{1}{4\sqrt{d}}
\Bigl|
C_{\mathscr{E}}\!\bigl(\boldsymbol{\hat{U}}(t)\boldsymbol{\hat{\rho}}\boldsymbol{\hat{U}}^{\dagger}(t)\bigr)
-
C_{\mathscr{E}}(\boldsymbol{\hat{\rho}})
\Bigr|,
\end{equation}
where $d:=\dim\bigl(\mathfrak{su}(\dim(\mathscr{H}))\bigr)$. Consequently,
\begin{equation}
\label{eq:coherence-lower-bound-sup}
\mathscr{G}_{HS}\!\bigl(\boldsymbol{\hat{U}}(t)\bigr)
\;\ge\;
\frac{1}{4\sqrt{d}}\,\mathscr{C}_{\mathscr{E}}\!\bigl(\boldsymbol{\hat{U}}(t)\bigr).
\end{equation}
In particular, the right-hand side controls the cohering and decohering power of
$\boldsymbol{\hat{U}}(t)$ and yields a basis-dependent operational lower bound for
$\mathscr{G}_{HS}(\boldsymbol{\hat{U}}(t))$.
\end{theorem}

\begin{proof}[Proof sketch]
The detailed proof is given in Appendix~\ref{sec:proofs-coherence}. The argument rests on two
observations: the derivative of $C_{\mathscr{E}}(\boldsymbol{\hat{\rho}}(t))$ along the unitary
trajectory can be written in terms of commutators involving $\mathscr{E}$ and
$\boldsymbol{\hat{H}}$; and this derivative is bounded in absolute value by
$\|\boldsymbol{\hat{H}}\|_{HS}$ via a Cauchy--Schwarz estimate. Integration over $[0,t]$ and use
of the Hilbert-Schmidt complexity formula $\mathscr{G}_{HS}(\boldsymbol{\hat{U}}(t))=
\frac{t}{\sqrt{d}}\|\boldsymbol{\hat{H}}\|_{HS}$ yield \eqref{eq:coherence-lower-bound}, and
taking the supremum over $\boldsymbol{\hat{\rho}}$ gives \eqref{eq:coherence-lower-bound-sup}.
\end{proof}

\subsection{Consequences for noise complexity}
\label{subsec:noise-coherence-consequence}

We now apply the coherence lower bound to the relative noise unitary
\begin{equation}
\label{eq:relative-noise-unitary-sec6}
\boldsymbol{\hat{U}}_{\mathrm{noise}}(t)
:=
\boldsymbol{\hat{U}}_{\mathrm{proj}}(t)\,
\boldsymbol{\hat{U}}_S(t)^\dagger,
\end{equation}
which measures the deviation between the implemented noisy projected dynamics and the ideal system
evolution. Its cohering power quantifies how strongly the noisy implementation departs from the
ideal dynamics in the chosen basis.

\begin{corollary}[Coherence-based lower bound for noise complexity]
\label{cor:omega-noise-coherence-bound}
Assume the penalty weights defining $\Omegahat$ satisfy $l_i\ge 1$. Then, for every initial state
$\boldsymbol{\hat{\rho}}\in\mathcal{S}(\mathscr{H}_{tot})$,
\begin{equation}
\label{eq:noise-coherence-lower-bound-state}
\mathcal{N}_{\Omegahat}(\Lambda_t;\mathfrak{D})
\;\ge\;
\frac{1}{4\sqrt{d}}
\Bigl|
C_{\mathscr{E}}\!\bigl(\boldsymbol{\hat{U}}_{\mathrm{noise}}(t)\boldsymbol{\hat{\rho}}\boldsymbol{\hat{U}}_{\mathrm{noise}}(t)^\dagger\bigr)
-
C_{\mathscr{E}}(\boldsymbol{\hat{\rho}})
\Bigr|,
\end{equation}
and consequently
\begin{equation}
\label{eq:noise-coherence-lower-bound-power}
\mathcal{N}_{\Omegahat}(\Lambda_t;\mathfrak{D})
\;\ge\;
\frac{1}{4\sqrt{d}}\,
\mathscr{C}_{\mathscr{E}}\!\bigl(\boldsymbol{\hat{U}}_{\mathrm{noise}}(t)\bigr),
\end{equation}
where $\mathscr{C}_{\mathscr{E}}$ denotes the cohering power of the relative noise unitary.
\end{corollary}

\begin{proof}
By right invariance of the Carnot--Carath\'eodory distance,
\begin{equation}
\mathcal{N}_{\Omegahat}(\Lambda_t;\mathfrak{D})
=
\mathscr{D}_{\Omegahat}\!\bigl(\boldsymbol{\hat{U}}_S(t),\boldsymbol{\hat{U}}_{\mathrm{proj}}(t)\bigr)
=
\mathscr{D}_{\Omegahat}\!\bigl(\boldsymbol{\hat{I}},\boldsymbol{\hat{U}}_{\mathrm{noise}}(t)\bigr).
\end{equation}
The assumption $l_i\ge 1$ and Proposition~\ref{prop:boundhsomega} give
$\mathcal{N}_{\Omegahat}(\Lambda_t;\mathfrak{D})\ge\mathscr{G}_{HS}(\boldsymbol{\hat{U}}_{\mathrm{noise}}(t))$.
Applying Theorem~\ref{thm:coherence-lower-bound} to $\boldsymbol{\hat{U}}_{\mathrm{noise}}(t)$
yields \eqref{eq:noise-coherence-lower-bound-state}, and taking the supremum over
$\boldsymbol{\hat{\rho}}$ gives \eqref{eq:noise-coherence-lower-bound-power}.
\end{proof}

This corollary shows that noise complexity is controlled not only by geometric distances from the
ideal trajectory, but also by the ability of the relative noise unitary to create coherence. The coherence-based lower bound complements the deficit-type bounds in
Section~\ref{sec:lindblad}: even when a direct complexity-deficit estimate becomes weak, the
relative noise unitary may still possess nontrivial cohering power, enforcing a nonzero lower
bound on $\mathcal{N}_{\Omegahat}(\Lambda_t;\mathfrak{D})$.

\section{Complexity Bounds in a Davies-Motivated GKSL Regime}
\label{sec:lindblad}

In this section we connect the geometric complexity functionals with the Markovian regime
described by GKSL (Lindblad) generators. The guiding physical picture is the weak-coupling limit
of a microscopic system-environment model, in the sense of Davies \cite{davis}: a GKSL semigroup
should be viewed not as carrying a canonical exact finite-dimensional dilation, but rather as the
asymptotic reduced dynamics of a weakly coupled unitary model over a suitable time scale. We
emphasize that the GKSL form of the master equation arises in the weak-coupling limit for both
bosonic and fermionic reservoir models \cite{BreuerPetruccione,CirioFermionic2022}.

\begin{remark}[Global versus local master equations]
\label{rem:global-vs-local}
The Davies weak-coupling derivation produces a \emph{global} (or \emph{secular}) master equation,
in which the secular approximation is applied after tracing over the environment in the energy
eigenbasis of the full system Hamiltonian $\boldsymbol{\hat{H}}_S$. This is the standard setting
of \cite{BreuerPetruccione,davis}, and all complexity bounds derived in this section should be
understood in this framework.

We remark that the \emph{local} master equation approach \cite{BreuerPetruccione}, in which
the secular approximation is omitted or applied subsystem-by-subsystem, also yields a GKSL
generator, but one that may differ from the global form at the level of coherences and, in
certain multi-partite settings, may not satisfy the second law of thermodynamics
\cite{LeviKosloff2014}. The geometric framework developed in this paper applies to \emph{any}
GKSL generator, regardless of whether it arises from a global or a local derivation: the
definitions of implementation-dependent complexity, intrinsic complexity, and noise complexity
in Sections~\ref{sec:defs-channel-noise}--\ref{sec:coherence-bounds} require only that a GKSL
generator $\mathcal{L}$ and its associated channel family $\Lambda_t = e^{t\mathcal{L}}$ be
given. The specific complexity bounds of Lemma~\ref{lem:gksl-two-sided-bound} and
Corollary~\ref{cor:gksl-dissipator-bounds} depend on the Davies weak-coupling ansatz
\eqref{eq:davies-ansatz} only through the bound $\|\boldsymbol{\hat{B}}_\alpha\|_{\mathrm{op}}\le\beta$
on the bath coupling operators; they hold equally for any dilation of the same structural form,
whether derived from a global or a local master equation.
\end{remark}

Motivated by this, we work on a fixed time window $[0,T]$ with a weak-coupling surrogate family
of dilations
\begin{equation}
\label{eq:davies-ansatz}
\boldsymbol{\hat{H}}_{\mathrm{tot}}^{(\lambda)}
=
\boldsymbol{\hat{H}}_S\otimes\boldsymbol{\hat{I}}_E
+
\lambda\,\boldsymbol{\hat{H}}_I
+
\boldsymbol{\hat{I}}_S\otimes\boldsymbol{\hat{H}}_E,
\qquad 0<\lambda\ll 1,
\end{equation}
whose reduced dynamics approximates the prescribed GKSL semigroup in the weak-coupling regime.
The estimates below are therefore model-dependent: their purpose is to extract structural
complexity bounds from the surrogate dilation model, rather than to assign an exact canonical
complexity to an abstract semigroup.

\subsection{GKSL generators and a dissipator scale}
\label{subsec:gksl-scale}

Let $(\Lambda_t)_{t\ge 0}$ be a quantum dynamical semigroup on $\mathscr{H}_S$ with GKSL
generator
\begin{equation}
\label{eq:gksl}
\frac{d}{dt}\Lambda_t(\boldsymbol{\hat{\rho}})
=
\mathcal{L}\!\bigl(\Lambda_t(\boldsymbol{\hat{\rho}})\bigr),
\qquad
\Lambda_0=\mathrm{Id},
\end{equation}
where
\begin{equation}
\label{eq:gksl-form}
\mathcal{L}(\boldsymbol{\hat{\rho}})
=
-i[\boldsymbol{\hat{H}}_S,\boldsymbol{\hat{\rho}}]
+
\sum_{\alpha=1}^{m}
\left(
\boldsymbol{\hat{L}}_\alpha \boldsymbol{\hat{\rho}}\,\boldsymbol{\hat{L}}_\alpha^\dagger
-\frac{1}{2}
\bigl\{
\boldsymbol{\hat{L}}_\alpha^\dagger \boldsymbol{\hat{L}}_\alpha,
\boldsymbol{\hat{\rho}}
\bigr\}
\right),
\end{equation}
with $\boldsymbol{\hat{H}}_S=\boldsymbol{\hat{H}}_S^\dagger$ and
$\{\boldsymbol{\hat{L}}_\alpha\}_{\alpha=1}^{m}\subset\mathcal{L}(\mathscr{H}_S)$.

A convenient coarse parameter controlling the dissipative scale is
\begin{equation}
\label{eq:dissipator-scale}
\Gamma
:=
\sum_{\alpha=1}^{m}
\|\boldsymbol{\hat{L}}_\alpha\|_{\mathrm{op}}^2,
\end{equation}
and the ideal reference evolution is
$\boldsymbol{\hat{U}}_S(t)=e^{-it\boldsymbol{\hat{H}}_S}$.

\subsection{Stability with respect to the effective generator}
\label{subsec:generator-stability}

The implementation-dependent channel complexity is defined directly from the effective generator
\begin{equation}
\label{eq:effective-generator}
\boldsymbol{\hat{K}}(t)
:=
\boldsymbol{\hat{H}}_{\mathrm{tot}}(t)
-
\boldsymbol{\hat{P}}_E\!\bigl(\boldsymbol{\hat{H}}_{\mathrm{tot}}(t)\bigr).
\end{equation}
The natural stability notion for this framework is stability with respect to perturbations of
$\boldsymbol{\hat{K}}(t)$ in the $L^1([0,t],\|\cdot\|_{\Omegahat})$ norm.

\begin{lemma}[Lipschitz stability of implementation-dependent channel complexity]
\label{lem:generator-stability}
Let $\mathfrak{D}_1$ and $\mathfrak{D}_2$ be two dilation implementations defined on the same
total Hilbert space, with effective generators
$\boldsymbol{\hat{K}}_j(t):=\boldsymbol{\hat{H}}_{\mathrm{tot}}^{(j)}(t)-\boldsymbol{\hat{P}}_E(\boldsymbol{\hat{H}}_{\mathrm{tot}}^{(j)}(t))$
for $j=1,2$. Then
\begin{equation}
\label{eq:generator-stability-bound}
\bigl|
\mathcal{G}^{\mathrm{imp}}_{\Omegahat}(\Lambda_t;\mathfrak{D}_1)
-
\mathcal{G}^{\mathrm{imp}}_{\Omegahat}(\Lambda_t;\mathfrak{D}_2)
\bigr|
\le
\frac{1}{\sqrt{\operatorname{Tr}(\Omegahat)}}
\int_0^t
\|\boldsymbol{\hat{K}}_1(s)-\boldsymbol{\hat{K}}_2(s)\|_{\Omegahat}\,ds.
\end{equation}
In particular, if $\int_0^t\|\boldsymbol{\hat{K}}_1(s)-\boldsymbol{\hat{K}}_2(s)\|_{\Omegahat}\,ds=O(\lambda t)$
as $\lambda\to 0$, then
$\mathcal{G}^{\mathrm{imp}}_{\Omegahat}(\Lambda_t;\mathfrak{D}_1)=\mathcal{G}^{\mathrm{imp}}_{\Omegahat}(\Lambda_t;\mathfrak{D}_2)+O(\lambda t)$.
\end{lemma}

\begin{proof}
By definition,
$\mathcal{G}^{\mathrm{imp}}_{\Omegahat}(\Lambda_t;\mathfrak{D}_j)=\frac{1}{\sqrt{\operatorname{Tr}(\Omegahat)}}\int_0^t\|\boldsymbol{\hat{K}}_j(s)\|_{\Omegahat}\,ds$.
The triangle inequality and the reverse triangle inequality for $\|\cdot\|_{\Omegahat}$ give
\eqref{eq:generator-stability-bound} immediately.
\end{proof}

\begin{remark}[Weak-coupling interpretation]
\label{rem:generator-stability-interpretation}
Lemma~\ref{lem:generator-stability} shows that channel complexity is stable under small
perturbations of the effective generator in the $L^1([0,t],\|\cdot\|_{\Omegahat})$ norm. In the
weak-coupling ansatz \eqref{eq:davies-ansatz}, the perturbation is
$\lambda(\boldsymbol{\hat{H}}_I-\boldsymbol{\hat{P}}_E(\boldsymbol{\hat{H}}_I))$, giving
$O(\lambda t)$ stability at leading order.
\end{remark}

\subsection{Two-sided weak-coupling bounds}
\label{subsec:gksl-channel-bounds}

We specialize \eqref{eq:davies-ansatz} to the standard GKSL-motivated interaction
\begin{equation}
\label{eq:standard-dilation-H}
\boldsymbol{\hat{H}}_I
=
\sum_{\alpha=1}^{m}
\left(
\boldsymbol{\hat{S}}_\alpha\otimes\boldsymbol{\hat{B}}_\alpha^\dagger
+
\boldsymbol{\hat{S}}_\alpha^\dagger\otimes\boldsymbol{\hat{B}}_\alpha
\right),
\qquad
\|\boldsymbol{\hat{B}}_\alpha\|_{\mathrm{op}}\le \beta.
\end{equation}
Here $\beta$ is a model-dependent finite constant that bounds the operator norm of the bath
coupling operators $\boldsymbol{\hat{B}}_\alpha$ within the chosen finite-dimensional surrogate.
Its value depends on the specific dilation model: in a finite-mode truncation of a bosonic
reservoir it is determined by the truncated mode occupation and spectral density, while in a
collision model it is set by the coupling strength per collision step \cite{BreuerPetruccione,Rivas2010}.
In either case, $\beta$ is finite by construction for any fixed truncation, but it may grow with
the dimension of the bath or depend on the so-called Bohr frequencies of the system Hamiltonian. The bounds
derived below should therefore be understood as holding for a fixed surrogate dilation with
given $\beta$, and not as model-independent estimates.
Since $\boldsymbol{\hat{P}}_E(\boldsymbol{\hat{I}}_S\otimes\boldsymbol{\hat{H}}_E)=\boldsymbol{\hat{I}}_S\otimes\boldsymbol{\hat{H}}_E$,
the purely environmental term is removed exactly by the subtraction, giving
\begin{equation}
\label{eq:effective-generator-lambda}
\boldsymbol{\hat{H}}_{\mathrm{tot}}^{(\lambda)}
-
\boldsymbol{\hat{P}}_E\!\bigl(\boldsymbol{\hat{H}}_{\mathrm{tot}}^{(\lambda)}\bigr)
=
\boldsymbol{\hat{H}}_S\otimes\boldsymbol{\hat{I}}_E
+
\lambda
\left(
\boldsymbol{\hat{H}}_I
-
\boldsymbol{\hat{P}}_E(\boldsymbol{\hat{H}}_I)
\right).
\end{equation}
For the $\Omegahat$-norm estimates there are natural constants:
$\kappa_{\Omegahat}(d_E)$ and $\chi_{\Omegahat}(d_S)$ such that
$\|\boldsymbol{\hat{A}}\otimes\boldsymbol{\hat{B}}\|_{\Omegahat}\le\kappa_{\Omegahat}(d_E)\|\boldsymbol{\hat{A}}\|_{\Omegahat}\|\boldsymbol{\hat{B}}\|_{\mathrm{op}}$
and
$\|\boldsymbol{\hat{X}}\|_{\Omegahat}\le\chi_{\Omegahat}(d_S)\|\boldsymbol{\hat{X}}\|_{\mathrm{op}}$
for all system operators; these are finite in any fixed finite-dimensional setting.

\begin{lemma}[Two-sided weak-coupling bounds]
\label{lem:gksl-two-sided-bound}
Assume $\mathfrak{D}^{(\lambda)}$ is time-independent on $[0,T]$, with
$d_{\mathrm{tot}}:=d_Sd_E$. Then, for every $t\in[0,T]$,
\begin{equation}
\label{eq:gksl-upper-bound-omega}
\mathcal{G}^{\mathrm{imp}}_{\Omegahat}(\Lambda_t;\mathfrak{D}^{(\lambda)})
\le
\frac{t}{\sqrt{d_{\mathrm{tot}}^2-1}}
\left(
\|\boldsymbol{\hat{H}}_S\otimes\boldsymbol{\hat{I}}_E\|_{\Omegahat}
+
2\lambda\beta\,\kappa_{\Omegahat}(d_E)
\sum_{\alpha=1}^{m}\|\boldsymbol{\hat{S}}_\alpha\|_{\Omegahat}
\right),
\end{equation}
and
\begin{equation}
\label{eq:gksl-lower-bound-omega}
\mathcal{G}^{\mathrm{imp}}_{\Omegahat}(\Lambda_t;\mathfrak{D}^{(\lambda)})
\ge
\frac{t}{\sqrt{d_{\mathrm{tot}}^2-1}}
\left[
\|\boldsymbol{\hat{H}}_S\otimes\boldsymbol{\hat{I}}_E\|_{\Omegahat}
-
2\lambda\beta\,\kappa_{\Omegahat}(d_E)
\sum_{\alpha=1}^{m}\|\boldsymbol{\hat{S}}_\alpha\|_{\Omegahat}
\right]_+,
\end{equation}
where $[x]_+:=\max\{x,0\}$.
\end{lemma}

\begin{proof}
Since $\mathfrak{D}^{(\lambda)}$ is time-independent,
$\mathcal{G}^{\mathrm{imp}}_{\Omegahat}(\Lambda_t;\mathfrak{D}^{(\lambda)})=\frac{t}{\sqrt{d_{\mathrm{tot}}^2-1}}\|\boldsymbol{\hat{H}}_{\mathrm{tot}}^{(\lambda)}-\boldsymbol{\hat{P}}_E(\boldsymbol{\hat{H}}_{\mathrm{tot}}^{(\lambda)})\|_{\Omegahat}$.
Applying the triangle inequality to \eqref{eq:effective-generator-lambda} with the bath bound
$\|\boldsymbol{\hat{B}}_\alpha\|_{\mathrm{op}}\le\beta$ and $\|\boldsymbol{\hat{A}}\otimes\boldsymbol{\hat{B}}\|_{\Omegahat}\le\kappa_{\Omegahat}(d_E)\|\boldsymbol{\hat{A}}\|_{\Omegahat}\|\boldsymbol{\hat{B}}\|_{\mathrm{op}}$
gives the upper bound. The lower bound follows from the reverse triangle inequality.
\end{proof}

\begin{corollary}[Dissipator-controlled weak-coupling bounds]
\label{cor:gksl-dissipator-bounds}
Under the assumptions of Lemma~\ref{lem:gksl-two-sided-bound},
\begin{equation}
\label{eq:gksl-upper-bound-gamma}
\mathcal{G}^{\mathrm{imp}}_{\Omegahat}(\Lambda_t;\mathfrak{D}^{(\lambda)})
\le
\frac{t}{\sqrt{d_{\mathrm{tot}}^2-1}}
\Bigl(
\|\boldsymbol{\hat{H}}_S\otimes\boldsymbol{\hat{I}}_E\|_{\Omegahat}
+
2\lambda\beta\,\kappa_{\Omegahat}(d_E)\chi_{\Omegahat}(d_S)\sqrt{m}\sqrt{\Gamma}
\Bigr),
\end{equation}
and
\begin{equation}
\label{eq:gksl-lower-bound-gamma}
\mathcal{G}^{\mathrm{imp}}_{\Omegahat}(\Lambda_t;\mathfrak{D}^{(\lambda)})
\ge
\frac{t}{\sqrt{d_{\mathrm{tot}}^2-1}}
\left[
\|\boldsymbol{\hat{H}}_S\otimes\boldsymbol{\hat{I}}_E\|_{\Omegahat}
-
2\lambda\beta\,\kappa_{\Omegahat}(d_E)\chi_{\Omegahat}(d_S)\sqrt{m}\sqrt{\Gamma}
\right]_+.
\end{equation}
In particular,
\begin{equation}
\label{eq:gksl-big-o-channel}
\mathcal{G}^{\mathrm{imp}}_{\Omegahat}(\Lambda_t;\mathfrak{D}^{(\lambda)})
=
\frac{t}{\sqrt{d_{\mathrm{tot}}^2-1}}
\|\boldsymbol{\hat{H}}_S\otimes\boldsymbol{\hat{I}}_E\|_{\Omegahat}
+
O(\lambda t),
\qquad
\lambda\to 0.
\end{equation}
\end{corollary}

\begin{proof}
Using $\|\boldsymbol{\hat{S}}_\alpha\|_{\Omegahat}\le\chi_{\Omegahat}(d_S)\|\boldsymbol{\hat{S}}_\alpha\|_{\mathrm{op}}$
and Cauchy--Schwarz
$\sum_\alpha\|\boldsymbol{\hat{S}}_\alpha\|_{\Omegahat}\le\chi_{\Omegahat}(d_S)\sqrt{m}\sqrt{\Gamma}$
in Lemma~\ref{lem:gksl-two-sided-bound} gives the bounds. The expansion \eqref{eq:gksl-big-o-channel} is immediate.
\end{proof}

\subsection{Consequences for intrinsic channel complexity}
\label{subsec:gksl-intrinsic-channel}

\begin{remark}[Intrinsic channel complexity in the weak-coupling regime]
\label{rem:gksl-intrinsic-channel}
Suppose the admissible resource model $\mathcal{R}$ is such that every
$\mathfrak{D}^{(\lambda)}\in\mathfrak{Dil}_{\mathcal{R}}^{[0,T]}(\Lambda_t)$ satisfies the
assumptions of Lemma~\ref{lem:gksl-two-sided-bound} with the same constants $\beta$,
$\kappa_{\Omegahat}(d_E)$, and $\chi_{\Omegahat}(d_S)$. Then
\eqref{eq:gksl-upper-bound-gamma}--\eqref{eq:gksl-lower-bound-gamma} hold uniformly over the
admissible class and therefore pass to the intrinsic channel complexity:
\begin{equation}
\label{eq:intrinsic-channel-upper-gksl}
\mathcal{G}_{\mathrm{intr}}^{\mathcal{R}}(\Lambda_t;[0,T])
\le
\frac{t}{\sqrt{d_{\mathrm{tot}}^2-1}}
\Bigl(
\|\boldsymbol{\hat{H}}_S\otimes\boldsymbol{\hat{I}}_E\|_{\Omegahat}
+
2\lambda\beta\,\kappa_{\Omegahat}(d_E)\chi_{\Omegahat}(d_S)\sqrt{m}\sqrt{\Gamma}
\Bigr),
\end{equation}
\begin{equation}
\label{eq:intrinsic-channel-lower-gksl}
\mathcal{G}_{\mathrm{intr}}^{\mathcal{R}}(\Lambda_t;[0,T])
\ge
\frac{t}{\sqrt{d_{\mathrm{tot}}^2-1}}
\left[
\|\boldsymbol{\hat{H}}_S\otimes\boldsymbol{\hat{I}}_E\|_{\Omegahat}
-
2\lambda\beta\,\kappa_{\Omegahat}(d_E)\chi_{\Omegahat}(d_S)\sqrt{m}\sqrt{\Gamma}
\right]_+,
\end{equation}
so that
$\mathcal{G}_{\mathrm{intr}}^{\mathcal{R}}(\Lambda_t;[0,T])=\frac{t}{\sqrt{d_{\mathrm{tot}}^2-1}}\|\boldsymbol{\hat{H}}_S\otimes\boldsymbol{\hat{I}}_E\|_{\Omegahat}+O(\lambda t)$.
\end{remark}

\subsection{Noise complexity and intrinsic noise complexity}
\label{subsec:gksl-noise}

For a given implementation $\mathfrak{D}^{(\lambda)}$, right invariance of the
Carnot--Carath\'eodory distance gives
\begin{equation}
\mathcal{N}_{\Omegahat}(\Lambda_t;\mathfrak{D}^{(\lambda)})
=
\mathscr{D}_{\Omegahat}\!\bigl(\boldsymbol{\hat{I}},\boldsymbol{\hat{U}}_{\mathrm{noise}}(t)\bigr),
\qquad
\boldsymbol{\hat{U}}_{\mathrm{noise}}(t)
:=
\boldsymbol{\hat{U}}_{\mathrm{proj}}(t)\boldsymbol{\hat{U}}_S(t)^\dagger.
\end{equation}

\begin{remark}[Implementation-dependent noise complexity in the weak-coupling regime]
\label{rem:gksl-noise-lower-bounds}
Two complementary lower bounds follow from the preceding identity. First, Corollary~\ref{cor:omega-noise-coherence-bound} applied to $\boldsymbol{\hat{U}}_{\mathrm{noise}}(t)$ gives
\begin{equation}
\label{eq:gksl-noise-coherence-lower}
\mathcal{N}_{\Omegahat}(\Lambda_t;\mathfrak{D}^{(\lambda)})
\ge
\frac{1}{4\sqrt{d}}\,
\mathscr{C}_{\mathscr{E}}\!\bigl(\boldsymbol{\hat{U}}_{\mathrm{noise}}(t)\bigr).
\end{equation}
Second, by the triangle inequality,
\begin{equation}
\label{eq:noise-deficit-bound}
\mathcal{N}_{\Omegahat}(\Lambda_t;\mathfrak{D}^{(\lambda)})
\ge
\max\!\Bigl\{
0,\,
\mathscr{G}_{\Omegahat}\!\bigl(\boldsymbol{\hat{U}}_S(t)\bigr)
-
\mathcal{G}^{\mathrm{imp}}_{\Omegahat}(\Lambda_t;\mathfrak{D}^{(\lambda)})
\Bigr\}.
\end{equation}
Combining these two bounds we get
\begin{equation}
\label{eq:noise-deficit-bound-gksl}
\mathcal{N}_{\Omegahat}(\Lambda_t;\mathfrak{D}^{(\lambda)})
\ge
\max\!\Bigl\{
\frac{1}{4\sqrt{d}}\,
\mathscr{C}_{\mathscr{E}}\!\bigl(\boldsymbol{\hat{U}}_{\mathrm{noise}}(t)\bigr),\,
\mathscr{G}_{\Omegahat}\!\bigl(\boldsymbol{\hat{U}}_S(t)\bigr)
-
\mathcal{G}^{\mathrm{imp}}_{\Omegahat}(\Lambda_t;\mathfrak{D}^{(\lambda)})
\Bigr\}
\end{equation}
The coherence-based bound is most informative outside the purely coherent-dominated regime, while the
deficit bound is strongest when the channel complexity falls significantly below the ideal
unitary complexity.
\end{remark}

\begin{remark}[Scope of the weak-coupling bounds and validity of the Markovian approximation]
\label{rem:gksl-scope}
The bounds in this section are intentionally coarse and model-dependent. Their purpose is to
show that, within a Davies-motivated weak-coupling surrogate, the implementation-dependent and
intrinsic channel complexities admit explicit two-sided control in terms of the Hamiltonian driving
and the dissipator scale $\Gamma$, while the noise complexities admit both coherence-based and
complexity-deficit lower bounds. Sharper estimates would require a more precise stability theory
for the complexity functionals under channel approximation, for instance in the diamond norm,
together with a detailed specification of the reservoir model and admissible $\Omegahat$.

Several physical caveats deserve explicit mention. First, the GKSL master equation is a
long-time approximation: it emerges from the microscopic unitary dynamics only in the
weak-coupling, Markovian limit, after coarse-graining over a timescale long compared to the
bath correlation time \cite{BreuerPetruccione,Rivas2010}. The surrogate dilation
\eqref{eq:davies-ansatz} therefore provides a valid approximation to the reduced dynamics only
for times $t$ in an intermediate window: long compared to the bath correlation time (so that
the Markovian approximation is valid) but short compared to any Poincar\'e recurrence time of
the finite truncated environment \cite{BocchieriLoinger1957,Peres1993}. As a consequence, the time window $[0,T]$ used throughout
this section should be understood as bounded above by this recurrence scale, and the complexity
bounds derived here inherit this limitation.

Second, collision models and finite-dimensional bath truncations generically display
non-Markovian behavior at short times: the generator does not grow linearly in $t$ for small
$t$, but only approaches linear growth in the long-time regime
\cite{Rivas2010,Cattaneo2021,Winczewski2024}. The scaling linear in $t$  of the complexity
bounds (e.g.\ \eqref{eq:gksl-upper-bound-gamma}) is therefore to be understood as a
long-time, Markovian-regime estimate, valid when $t$ is sufficiently large relative to the
bath correlation time.

These limitations are intrinsic to the Davies-type surrogate model and do not affect the
abstract framework of Sections~\ref{sec:defs-channel-noise}--\ref{sec:coherence-bounds}, which
applies to any time-dependent family of channels regardless of whether it arises from a
Markovian master equation.
\end{remark}

\section{Complexity Benchmarks for Canonical Noise Channels}
\label{sec:benchmarks}

In this section we benchmark the proposed functionals on canonical single-qubit noise models within
the Davies-motivated weak-coupling framework of Section~\ref{sec:lindblad}. The purpose of these
examples is to illustrate how the implementation-dependent channel complexity and its associated
noise contributions behave in physically standard Markovian regimes. 

Before continuting note that the constants $\kappa_{\Omegahat}(d_E)$ and
$\chi_{\Omegahat}(d_S)$ may therefore be chosen by combining this norm
equivalence with the standard bounds
$\|\boldsymbol{\hat A}\otimes\boldsymbol{\hat B}\|_{\mathrm{HS}}
=
\|\boldsymbol{\hat A}\|_{\mathrm{HS}}\|\boldsymbol{\hat B}\|_{\mathrm{HS}}$
and
$\|\boldsymbol{\hat X}\|_{\mathrm{HS}}\le \sqrt{d_S}\,\|\boldsymbol{\hat X}\|_{\mathrm{op}}$

\paragraph{Davies weak-coupling framework and effective rates.}
Recall that the microscopic surrogate dynamics is generated by
$\boldsymbol{\hat{H}}_{\mathrm{tot}}^{(\lambda)}=\boldsymbol{\hat{H}}_S\otimes\boldsymbol{\hat{I}}_E+\lambda\,\boldsymbol{\hat{H}}_I+\boldsymbol{\hat{I}}_S\otimes\boldsymbol{\hat{H}}_E$
with interaction Hamiltonian
\begin{equation}
\label{eq:bench-HI}
\boldsymbol{\hat{H}}_I
=
\sum_{\alpha}
\left(
\boldsymbol{\hat{S}}_\alpha\otimes \boldsymbol{\hat{B}}_\alpha^\dagger
+
\boldsymbol{\hat{S}}_\alpha^\dagger\otimes \boldsymbol{\hat{B}}_\alpha
\right).
\end{equation}
In the Davies weak-coupling limit (after time rescaling $t\mapsto t/\lambda^2$), the reduced
dynamics is described by a GKSL generator with effective Lindblad operators
$\boldsymbol{\hat{L}}_\alpha$, obtained from the Bohr-frequency components of the microscopic
couplings $\boldsymbol{\hat{S}}_\alpha$ dressed by bath correlation functions. In particular, the
dissipative rates satisfy $\gamma_\alpha = \mathcal{O}(\lambda^2)$, so that
$\boldsymbol{\hat{L}}_\alpha \sim \sqrt{\gamma_\alpha}\,\boldsymbol{\hat{S}}_\alpha$ at the level
of scaling. In the benchmark analysis below, we treat the rates $\gamma,\kappa,\gamma_j$ as
effective parameters encoding the underlying microscopic coupling, and express the complexity
bounds using Lemma~\ref{lem:gksl-two-sided-bound}. The symbol $\lesssim$ denotes an inequality
valid up to model-dependent constants within the weak-coupling approximation. As discussed in
Remark~\ref{rem:gksl-scope}, the time window $[0,T]$ is implicitly bounded above by the
Poincar\'e recurrence time of the surrogate bath, so all estimates below should be understood as
holding in the intermediate Markovian regime where the GKSL approximation is valid
\cite{Rivas2010,Winczewski2024}.

Throughout we fix $d_S=2$ and take as ideal (closed) reference evolution
\begin{equation}
\label{eq:bench-ideal-unitary}
\boldsymbol{\hat{U}}_S(t)
=
e^{-it\boldsymbol{\hat{H}}_S},
\qquad
\boldsymbol{\hat{H}}_S
=
\frac{\omega}{2}\,\boldsymbol{\hat{\sigma}}_z.
\end{equation}
The corresponding noiseless unitary geometric complexity in the Hilbert-Schmidt specialization is
\begin{equation}
\label{eq:bench-ideal-complexity}
\mathscr{G}_{HS}\!\bigl(\boldsymbol{\hat{U}}_S(t)\bigr)
=
\frac{t}{\sqrt{d_S^2-1}}
\|\boldsymbol{\hat{H}}_S\|_{HS}
=
\frac{t}{\sqrt{3}}\cdot \frac{|\omega|}{2}\sqrt{2}.
\end{equation}
For the noisy dynamics, the implementation-dependent channel complexity is controlled by the
projected generator
$\boldsymbol{\hat{H}}_{\mathrm{eff}}=\boldsymbol{\hat{H}}_{\mathrm{tot}}-\boldsymbol{\hat{P}}_E(\boldsymbol{\hat{H}}_{\mathrm{tot}})$
(see Definition~\ref{eq:implementation-dependent-channel-complexity}).

\paragraph{Standard dilation: explicit values of $d_E$ and $\beta$.}
In all three benchmarks below we adopt the \emph{standard single-ancilla dilation}, in which
the environment is a single qubit, $d_E=2$, giving
$d_{\mathrm{tot}}=d_Sd_E=4$.
The bath coupling operators $\boldsymbol{\hat{B}}_\alpha$ are chosen from the set
$\{\boldsymbol{\hat{\sigma}}_x, \boldsymbol{\hat{\sigma}}_y, \boldsymbol{\hat{\sigma}}_z,
\boldsymbol{\hat{\sigma}}_+, \boldsymbol{\hat{\sigma}}_-\}$, all of which satisfy
$\|\boldsymbol{\hat{B}}_\alpha\|_{\mathrm{op}}=1$.
We therefore fix
\begin{equation}
\label{eq:bench-standard-params}
d_E = 2,
\qquad
\beta = 1,
\qquad
d_{\mathrm{tot}} = 4.
\end{equation}
In the Hilbert-Schmidt specialization $\Omegahat=\boldsymbol{\hat{I}}$, the norm constants
reduce to $\kappa_{HS}(d_E)=1$ and $\chi_{HS}(d_S)=1$ (since
$\|\boldsymbol{\hat{A}}\otimes\boldsymbol{\hat{B}}\|_{HS}=\|\boldsymbol{\hat{A}}\|_{HS}\|\boldsymbol{\hat{B}}\|_{HS}$),
and the bounds of Lemma~\ref{lem:gksl-two-sided-bound} become fully explicit with no hidden
model-dependent constants beyond the effective rates $\gamma,\kappa,\gamma_j$.
With these choices, $d_{\mathrm{tot}}^2-1=15$, so all upper bounds carry the overall
prefactor $t/\sqrt{15}$.

\subsection{Pure dephasing (phase damping)}
\label{subsec:bench-dephasing}

\paragraph{Model and GKSL form:}
A pure-dephasing model arises from a microscopic coupling with system operator
$\boldsymbol{\hat{S}}=\boldsymbol{\hat{\sigma}}_z$. In the Davies limit this yields
\begin{equation}
\label{eq:bench-dephasing-gksl}
\mathcal{L}_{\mathrm{deph}}(\boldsymbol{\hat{\rho}})
=
-i[\boldsymbol{\hat{H}}_S,\boldsymbol{\hat{\rho}}]
+
\frac{\gamma}{2}
\left(
\boldsymbol{\hat{\sigma}}_z \boldsymbol{\hat{\rho}}\,\boldsymbol{\hat{\sigma}}_z 
-
\boldsymbol{\hat{\rho}}
\right),
\end{equation}
with effective Lindblad operator $\boldsymbol{\hat{L}}=\sqrt{\gamma/2}\, \boldsymbol{\sigma}_z$.

\paragraph{Complexity estimate:}
Using $\|\boldsymbol{\hat{\sigma}}_z \|_{HS}=\sqrt{2}$, Lemma~\ref{lem:gksl-two-sided-bound} yields
\begin{equation}
\label{eq:bench-dephasing-Gbound}
\mathcal{G}_{HS}^{\mathrm{imp}}\!\bigl(\Lambda_t^{\mathrm{deph}};\mathfrak{D}\bigr)
\ \lesssim\
\frac{t}{\sqrt{d_{\mathrm{tot}}^2-1}}
\left(
\|\boldsymbol{\hat{H}}_S\|_{HS}\sqrt{d_E}
+
2\beta\,\sqrt{\gamma}
\right).
\end{equation}
Substituting the standard-dilation values \eqref{eq:bench-standard-params}
($d_E=2$, $\beta=1$, $d_{\mathrm{tot}}=4$) and $\|\boldsymbol{\hat{H}}_S\|_{HS}=|\omega|/\sqrt{2}$,
this becomes fully explicit:
\begin{equation}
\label{eq:bench-dephasing-Gbound-explicit}
\mathcal{G}_{HS}^{\mathrm{imp}}\!\bigl(\Lambda_t^{\mathrm{deph}};\mathfrak{D}\bigr)
\ \lesssim\
\frac{t}{\sqrt{15}}
\Bigl(|\omega| + 2\sqrt{\gamma}\Bigr).
\end{equation}

\paragraph{Interpretation:}
Pure dephasing illustrates the distinction between coherent distortion and full
implementation-dependent complexity: the system-local coherent motion may be trivial, while the
interaction contribution to $\boldsymbol{\hat{H}}_{\mathrm{eff}}$ remains nonzero. Indeed, as
shown in Appendix~\ref{app:dephasing-zero-coherent-distortion}, the coherent distortion complexity
$\mathcal{G}^{\mathrm{coh}}_{\Omegahat}$ vanishes for the standard dephasing Stinespring dilation, even
though the channel is physically nontrivial.

\paragraph{Noise complexity trend:}
The dissipative contribution scales as $\sqrt{\gamma}\,t$, so increasing $\gamma$ increases the
deviation from the ideal trajectory.

\subsection{Amplitude damping (energy relaxation)}
\label{subsec:bench-amplitude-damping}

\paragraph{Model and GKSL form:}
Amplitude damping arises from an exchange-type microscopic coupling with
$\boldsymbol{\hat{S}}=\boldsymbol{\hat{\sigma}}_-$. The Davies limit yields
\begin{equation}
\label{eq:bench-ad-gksl}
\mathcal{L}_{\mathrm{AD}}(\boldsymbol{\hat{\rho}})
=
-i[\boldsymbol{\hat{H}}_S,\boldsymbol{\hat{\rho}}]
+
\kappa
\left(
\boldsymbol{\hat{\sigma}}_- \boldsymbol{\hat{\rho}}\,\boldsymbol{\hat{\sigma}}_+
-
\frac{1}{2}\{\boldsymbol{\hat{\sigma}}_+\boldsymbol{\hat{\sigma}}_-,\boldsymbol{\hat{\rho}}\}
\right),
\end{equation}
with $\boldsymbol{\hat{L}}=\sqrt{\kappa}\,\boldsymbol{\hat{\sigma}}_-$.

\paragraph{Complexity estimate:}
Using $\|\boldsymbol{\hat{\sigma}}_-\|_{HS}=1$, Lemma~\ref{lem:gksl-two-sided-bound} gives
\begin{equation}
\label{eq:bench-ad-Gbound}
\mathcal{G}_{HS}^{\mathrm{imp}}\!\bigl(\Lambda_t^{\mathrm{AD}};\mathfrak{D}\bigr)
\ \lesssim\
\frac{t}{\sqrt{d_{\mathrm{tot}}^2-1}}
\left(
\|\boldsymbol{\hat{H}}_S\|_{HS}\sqrt{d_E}
+
2\beta\,\sqrt{\kappa}
\right).
\end{equation}
Substituting \eqref{eq:bench-standard-params} ($d_E=2$, $\beta=1$, $d_{\mathrm{tot}}=4$)
and $\|\boldsymbol{\hat{H}}_S\|_{HS}=|\omega|/\sqrt{2}$, the bound becomes
\begin{equation}
\label{eq:bench-ad-Gbound-explicit}
\mathcal{G}_{HS}^{\mathrm{imp}}\!\bigl(\Lambda_t^{\mathrm{AD}};\mathfrak{D}\bigr)
\ \lesssim\
\frac{t}{\sqrt{15}}
\Bigl(|\omega| + 2\sqrt{\kappa}\Bigr).
\end{equation}

\paragraph{Interpretation:}
This model captures both population relaxation and decoherence, providing a benchmark in which
coherent and dissipative contributions to $\boldsymbol{\hat{H}}_{\mathrm{eff}}$ are simultaneously
present. Unlike pure dephasing, amplitude damping induces nonzero coherent distortion complexity
because $\boldsymbol{\hat{\sigma}}_-$ has a nonzero projection onto the system-local subspace.

\paragraph{Noise complexity trend:}
The dissipative contribution scales as $\sqrt{\kappa}\,t$, so stronger relaxation increases the
deviation from the ideal unitary evolution.

\subsection{Depolarizing and Pauli channels}
\label{subsec:bench-depolarizing}

\paragraph{Model and GKSL form.}
For isotropic Pauli noise, one takes three coupling directions
$\boldsymbol{\hat{S}}_j=\boldsymbol{\hat{\sigma}}_j$. The Davies limit yields
\begin{equation}
\label{eq:bench-depol-gksl}
\mathcal{L}_{\mathrm{dep}}(\boldsymbol{\hat{\rho}})
=
-i[\boldsymbol{\hat{H}}_S,\boldsymbol{\hat{\rho}}]
+
\frac{\gamma}{2}
\sum_{j\in\{x,y,z\}}
\left(
\boldsymbol{\hat{\sigma}}_j \boldsymbol{\hat{\rho}}\,\boldsymbol{\hat{\sigma}}_j
-
\boldsymbol{\hat{\rho}}
\right),
\end{equation}
with Lindblad operators $\boldsymbol{\hat{L}}_j=\sqrt{\gamma/2}\,\boldsymbol{\hat{\sigma}}_j$.

\paragraph{Complexity estimate:}
Since $\|\boldsymbol{\hat{\sigma}}_j\|_{HS}=\sqrt{2}$, Lemma~\ref{lem:gksl-two-sided-bound} gives
\begin{equation}
\label{eq:bench-depol-Gbound}
\mathcal{G}_{HS}^{\mathrm{imp}}\!\bigl(\Lambda_t^{\mathrm{dep}};\mathfrak{D}\bigr)
\ \lesssim\
\frac{t}{\sqrt{d_{\mathrm{tot}}^2-1}}
\left(
\|\boldsymbol{\hat{H}}_S\|_{HS}\sqrt{d_E}
+
6\beta\,\sqrt{\gamma}
\right).
\end{equation}
Substituting \eqref{eq:bench-standard-params} ($d_E=2$, $\beta=1$, $d_{\mathrm{tot}}=4$)
and $\|\boldsymbol{\hat{H}}_S\|_{HS}=|\omega|/\sqrt{2}$, this reads
\begin{equation}
\label{eq:bench-depol-Gbound-explicit}
\mathcal{G}_{HS}^{\mathrm{imp}}\!\bigl(\Lambda_t^{\mathrm{dep}};\mathfrak{D}\bigr)
\ \lesssim\
\frac{t}{\sqrt{15}}
\Bigl(|\omega| + 6\sqrt{\gamma}\Bigr).
\end{equation}
For the anisotropic Pauli semigroup with rates $\gamma_j$ one analogously obtains
\begin{equation}
\label{eq:bench-pauli-aniso-Gbound}
\mathcal{G}_{HS}^{\mathrm{imp}}\!\bigl(\Lambda_t^{\mathrm{Pauli}};\mathfrak{D}\bigr)
\ \lesssim\
\frac{t}{\sqrt{d_{\mathrm{tot}}^2-1}}
\left(
\|\boldsymbol{\hat{H}}_S\|_{HS}\sqrt{d_E}
+
2\beta\sum_{j\in\{x,y,z\}}\sqrt{\gamma_j}
\right),
\end{equation}
which with the standard-dilation values becomes
\begin{equation}
\label{eq:bench-pauli-aniso-Gbound-explicit}
\mathcal{G}_{HS}^{\mathrm{imp}}\!\bigl(\Lambda_t^{\mathrm{Pauli}};\mathfrak{D}\bigr)
\ \lesssim\
\frac{t}{\sqrt{15}}
\left(|\omega|
+
2\sum_{j\in\{x,y,z\}}\sqrt{\gamma_j}
\right).
\end{equation}

\paragraph{Interpretation:}
Depolarizing noise involves three independent dissipative directions that contribute additively to
the interaction sector of $\boldsymbol{\hat{H}}_{\mathrm{eff}}$, leading to a larger overall
implementation cost than in single-channel models at the same noise rate.

\paragraph{Noise complexity trend:}
The noise complexity grows more rapidly than in single-channel models, reflecting the accumulation
of independent dissipative contributions.

\medskip

\paragraph{Implementation versus intrinsic complexity:}
The above estimates are expressed in terms of the implementation-dependent channel complexity for
a specific dilation within the Davies weak-coupling framework. As discussed in
Section~\ref{sec:defs-channel-noise}, the intrinsic channel complexity is defined by minimizing
over an admissible class of dilations. From this perspective, the Davies construction selects a
physically motivated admissible family of weak-coupling implementations. The bounds obtained here
therefore provide upper bounds on the intrinsic channel complexity relative to this class:
\begin{equation}
\mathcal{G}_{\mathrm{intr}}^{\mathcal{R}}(\Lambda_t)
\;\le\;
\mathcal{G}_{HS}^{\mathrm{imp}}(\Lambda_t;\mathfrak{D}_{\mathrm{Davies}}).
\end{equation}

\begin{remark}[Difficulty of computing the exact intrinsic complexity]
\label{rem:intrinsic-difficulty}
Computing $\mathcal{G}_{\mathrm{intr}}^{\mathcal{R}}(\Lambda_t)$ exactly --- rather than
bounding it from above by a specific dilation --- requires solving the variational problem
\begin{equation}
\mathcal{G}_{\mathrm{intr}}^{\mathcal{R}}(\Lambda_t;[0,T])
=
\inf_{\mathfrak{D}\,\in\,\mathfrak{Dil}_{\mathcal{R}}^{[0,T]}(\Lambda_t)}
\mathcal{G}_{HS}^{\mathrm{imp}}(\Lambda_t;\mathfrak{D}),
\end{equation}
and this is genuinely hard for the following reasons. First, the admissible set
$\mathfrak{Dil}_{\mathcal{R}}^{[0,T]}(\Lambda_t)$ is defined by the resource constraints
\eqref{eq:resource_constraints} together with the dynamical consistency condition
\eqref{eq:channel-from-dilation}; the latter is a nonlinear, operator-valued constraint
linking the time-ordered exponential of $\boldsymbol{\hat{H}}_{\mathrm{tot}}(t)$ to the
target channel family $\Lambda_t$. In general, no closed-form parametrization of this
constraint set is available.

Second, even when the admissible class is restricted to the standard single-ancilla weak-coupling
family \eqref{eq:davies-ansatz}, the optimization over the interaction Hamiltonian
$\boldsymbol{\hat{H}}_I$ subject to the consistency constraint is a nonconvex problem in the
space of traceless operators on $\mathscr{H}_S\otimes\mathscr{H}_E$.

In the benchmark examples of this section, the exact intrinsic complexity is therefore not
computed. The standard Davies dilation provides a natural and physically motivated
representative of the admissible class, and the resulting implementation-dependent bounds serve
as explicit upper bounds on the intrinsic quantity. For the pure-dephasing channel, the
situation is partially tractable: the dilation used in
Appendix~\ref{app:dephasing-zero-coherent-distortion} is in fact the minimal-dimension
realization ($d_E=2$), so any further reduction in $d_E$ is impossible within the qubit
setting, and the standard dilation bound is already tight with respect to the dimension
resource. Whether it achieves the true minimum of
$\mathcal{G}_{HS}^{\mathrm{imp}}$ over all couplings with $d_E=2$ remains an open question
that would require a full optimization over the coupling operators $\boldsymbol{\hat{H}}_I$.
\end{remark}

\paragraph{Noise complexity:}
Within the weak-coupling regime, the deviation from the ideal unitary trajectory is controlled by
the dissipative contribution, which scales as $\sqrt{\gamma}\,t$ (or $\sqrt{\kappa}\,t$). Thus,
noise complexity grows at most linearly in time and increases with the square root of the
effective dissipative rates.

\paragraph{$\Omegahat$-weighted geometries.}
While the benchmarks above are presented in the Hilbert-Schmidt case $\Omegahat=\boldsymbol{\hat{I}}$,
the same analysis extends directly to the general $\Omegahat$-weighted geometries of
Section~\ref{sec:unitary_complexity}. In that setting, the bounds are obtained by replacing
Hilbert-Schmidt norms with the corresponding $\Omegahat$-dependent norms, leading to analogous
scaling behavior with model-dependent weighting of different operator directions.

\section{Conclusions and Final Remarks}
\label{sec:conclusions}

We have introduced and analyzed a geometric framework for extending Nielsen-type circuit
complexity from closed unitary dynamics to open-system evolutions modeled by quantum channels.
The central structural difficulty is the non-uniqueness of microscopic realizations: the same
reduced channel $\Lambda_t$ may arise from many inequivalent Stinespring dilations, differing in
environment size, preparation, and coupling. To make the notion of implementation cost
mathematically meaningful, we therefore distinguished two layers: an implementation-dependent
channel complexity $\mathcal{G}^{\mathrm{imp}}_{\Omegahat}(\Lambda_t;\mathfrak{D})$, assigned to a specified dilation
$\mathfrak{D}$, and an intrinsic channel complexity obtained by minimizing over an admissible
class of dilations relative to a chosen resource model.

A key feature of the framework is the subtractive structure in the implementation-dependent
definition. Rather than introducing the subtraction term by hand, we derived it from a compact
set of geometric and operational postulates: consistency with the closed-system limit, neutrality
under purely environmental evolution, and stability under the natural dilation gauge. This
produces a channel-complexity functional that removes contributions which are purely environmental
while retaining sensitivity to system-local and interaction terms that are physically relevant to
the reduced dynamics.

At the same time, the framework makes explicit that channel complexity is inherently
resource-relative. The admissible class of dilations specifies which microscopic implementations
are allowed, while the penalty tensor $\Omegahat$ specifies how different operator directions are
weighted geometrically. In this sense, $\Omegahat$ is not merely a background metric but part of the
resource structure itself: different choices of $\Omegahat$ correspond to different notions of what
it means for an implementation to be costly.

Alongside channel complexity, we introduced a companion quantity, noise complexity, which measures
the geometric deviation of a noisy implementation from a prescribed ideal closed-system evolution.
This distinction between implementation cost and deviation from ideality is, in our view, one of
the main conceptual advantages of the framework: channel complexity asks how costly a noisy
realization is, while noise complexity asks how much geometric correction would be required to
bring it back to the target evolution.

On the technical side, we established an operational coherence-based lower bound for unitary
geometric complexity in the Hilbert-Schmidt specialization. This shows that the
cohering power of a unitary relative to a fixed basis controls from below its
geometric cost. By transferring the same idea to the relative noise unitary, we obtained
coherence-based lower bounds for noise complexity that complement the deficit-type estimates
arising in the weak-coupling analysis. This provides an operational interpretation of both
unitary and noise complexity and clarifies how geometric cost can grow with the dimension of the
system through the increasing richness of possible coherent dynamics.

We also derived structural properties of the channel functional, including linear time-scaling
under time-homogeneous dilation models, and explicit weak-coupling bounds in the Markovian GKSL
regime. These results should be understood as coarse but physically meaningful estimates within a
Davies-motivated surrogate model. In particular, the benchmark examples of dephasing, amplitude
damping, and Pauli/depolarizing noise show how the framework distinguishes coherent distortion
from full implementation cost, how the weak-coupling picture selects a physically motivated
admissible class of dilations, and how noise complexity tracks qualitative degradation relative
to the ideal closed evolution.

Several directions remain open. First, beyond the Hilbert-Schmidt and particular $\Omegahat$  specialization in this paper, it would be
important to develop the theory of anisotropic $\Omegahat$-weighted geometries in more detail,
especially when $\Omegahat$ encodes locality penalties, interaction penalties, or other
resource-sensitive structures. Second, the intrinsic minimization over admissible dilations raises
natural questions of well-posedness, attainability, and dependence on the chosen resource model.
As discussed in Remark~\ref{rem:intrinsic-difficulty}, computing the exact intrinsic complexity
requires solving a nonconvex variational problem over the set of dilations consistent with the
target channel family, which is analytically tractable only in special cases (such as the
minimal-dimension dephasing dilation, where the environment cannot be further reduced). Developing
systematic methods --- whether variational, control-theoretic, or based on convex relaxations ---
for the intrinsic minimization problem is an important open direction.
Third, in the Markovian setting, sharper stability and approximation results would be needed to
pass from Davies-type surrogate bounds to more precise statements about physical channels,
possibly in stronger norms such as the diamond norm. Finally, it would be interesting to connect
the present geometric framework more systematically with control-theoretic notions of optimal synthesis and with other resource theories of noisy quantum processes.

More broadly, the perspective developed here suggests that the geometry of complexity for open
quantum systems should not be sought directly at the level of reduced channels alone, but rather
through the interplay between microscopic realization, admissible resources, and the quotient
structures induced by environmental gauge freedom. In that sense, the complexity of a quantum
channel is not simply a property of the channel as an abstract CPTP map, but of the physically
meaningful ways in which that channel can be implemented.

\appendix

\section{Vanishing coherent distortion complexity for the dephasing channel}
\label{app:dephasing-zero-coherent-distortion}

In this appendix we show explicitly that the coherent distortion complexity introduced in
Definition~\ref{eq:coherent-distortion-complexity} vanishes for a standard Stinespring
realization of the qubit dephasing channel. This illustrates that $\mathcal{G}^{\mathrm{coh}}_{\Omegahat}$
detects only the subsystem-visible coherent component of the dilation generator, and not purely
decohering effects.

\subsection{Dephasing channel and Stinespring dilation}

Consider the one-parameter family of qubit dephasing channels
\begin{equation}
\Lambda_{\theta}(\boldsymbol{\hat{\rho}})
=
\cos^2\theta\,\boldsymbol{\hat{\rho}}
+
\sin^2\theta\,\boldsymbol{\hat{\sigma}}_{z}^{S}\boldsymbol{\hat{\rho}}\boldsymbol{\hat{\sigma}}_{z}^{S},
\qquad
\theta\in[0,\pi/2].
\end{equation}
A convenient Stinespring dilation introduces a qubit environment initially prepared in
$\ket{0}_E$ and defines the joint unitary
$\boldsymbol{\hat{U}}_{SE}(\theta)=e^{-i\theta\,\boldsymbol{\hat{\sigma}}_{z}^{S}\otimes\boldsymbol{\hat{\sigma}}_{y}^{E}}$.
Since $e^{-i\theta\,\boldsymbol{\hat{\sigma}}_{z}^{S}\otimes\boldsymbol{\hat{\sigma}}_{y}^{E}}=\cos\theta\,\boldsymbol{\hat{I}}_{SE}-i\sin\theta\,\boldsymbol{\hat{\sigma}}_{z}^{S}\otimes\boldsymbol{\hat{\sigma}}_{y}^{E}$
and $\boldsymbol{\hat{\sigma}}_{y}^{E}\ket{0}_E=i\ket{1}_E$, one verifies that
\begin{equation}
\mathrm{Tr}_E\!\left[
\boldsymbol{\hat{U}}_{SE}(\theta)
(\boldsymbol{\hat{\rho}}\otimes\ket{0}\!\bra{0}_E)
\boldsymbol{\hat{U}}_{SE}(\theta)^\dagger
\right]
=
\cos^2\theta\,\boldsymbol{\hat{\rho}}+\sin^2\theta\,\boldsymbol{\hat{\sigma}}_{z}\boldsymbol{\hat{\rho}} \boldsymbol{\hat{\sigma}}_{z}.
\end{equation}

\subsection{Right-trivialized generator}

Let $s\in[0,\theta]$. The right-trivialized generator of
$\boldsymbol{\hat{U}}_{SE}(s)=e^{-is\,\boldsymbol{\hat{\sigma}}_{z}^{S}\otimes\boldsymbol{\hat{\sigma}}_{y}^{E}}$ is
\begin{equation}
\boldsymbol{\hat{H}}_{\mathrm{tot}}(s)
:=
i\,\boldsymbol{\dot{\hat{U}}}_{SE}(s)\boldsymbol{\hat{U}}^{\dagger}_{SE}(s)
=
\boldsymbol{\hat{\sigma}}_{z}^{S}\otimes \boldsymbol{\hat{\sigma}}_{y}^{E}.
\end{equation}
Thus the dilation is generated entirely by an interaction term. The system-local projection is
\begin{equation}
\boldsymbol{\hat{P}}_S(\boldsymbol{\hat{H}}_{\mathrm{tot}}(s))
=
\frac{1}{d_E}\,\mathrm{Tr}_E\!\bigl(\boldsymbol{\hat{\sigma}}_{z}^{S}\otimes \boldsymbol{\hat{\sigma}}_{y}^{E}\bigr)\otimes \boldsymbol{\hat{I}}_E
=
\frac{1}{2}\,\boldsymbol{\hat{\sigma}}_{z}^{S}\,\mathrm{Tr}(\boldsymbol{\hat{\sigma}}_{y}^{E})\otimes \boldsymbol{\hat{I}}_E
=
0,
\end{equation}
since $\mathrm{Tr}(\boldsymbol{\hat{\sigma}}_{y}^{E})=0$.

\subsection{Vanishing coherent distortion complexity}

Substituting into Definition~\ref{eq:coherent-distortion-complexity},
\begin{equation}
\mathcal{G}^{\mathrm{coh}}_{\Omegahat}(\Lambda_{\theta};\mathfrak{D})
=
\frac{1}{\sqrt{\operatorname{Tr}(\Omegahat)}}
\int_0^\theta
\left\|
\boldsymbol{\hat{P}}_S\!\bigl(\boldsymbol{\hat{H}}_{\mathrm{tot}}(s)\bigr)
\right\|_{\Omegahat}\, ds
=
0.
\end{equation}

\subsection{Interpretation}

The vanishing of $\mathcal{G}^{\mathrm{coh}}_{\Omegahat}$ does not mean that the dephasing channel is
physically trivial. Rather, the Stinespring generator $\boldsymbol{\hat{\sigma}}_{z}^{S}\otimes\boldsymbol{\hat{\sigma}}_{y}^{E}$
is a purely interaction-type term whose projection onto the system-local subspace
$\mathfrak{su}(d_S)\otimes\boldsymbol{\hat{I}}_E$ vanishes. Consequently, $\mathcal{G}^{\mathrm{coh}}_{\Omegahat}$
is insensitive to this noise process, consistent with its interpretation as a measure of
\emph{subsystem-visible coherent geometric distortion}. This observation motivates the broader
implementation-dependent complexity of Definition~\ref{eq:implementation-dependent-channel-complexity},
which retains interaction terms while quotienting out only the purely environmental contribution.

\section{Proof of the coherence-power lower bound}
\label{sec:proofs-coherence}

This appendix collects the technical steps underlying the coherence-based lower bound stated as
Theorem~\ref{thm:coherence-lower-bound}. Throughout, $\mathscr{E}$ is the dephasing map defined
in Section~\ref{sec:coherence-bounds} and $C_{\mathscr{E}}$ is the linear-entropy coherence
\eqref{def:entropycoherencemeasure}:
\begin{equation}
C_{\mathscr{E}}(\boldsymbol{\hat{\rho}})
=
S_L\!\bigl(\mathscr{E}(\boldsymbol{\hat{\rho}})\bigr)-S_L(\boldsymbol{\hat{\rho}})
=
\Tr(\boldsymbol{\hat{\rho}}^2)-\Tr\!\bigl(\mathscr{E}(\boldsymbol{\hat{\rho}})^2\bigr).
\end{equation}

\subsection*{B.1. Dephasing as an orthogonal projector and a norm identity}

\begin{lemma}[Dephasing is an orthogonal projector in Hilbert-Schmidt geometry]
\label{lem:dephasing-projector}
The dephasing map $\mathscr{E}$ is a self-adjoint idempotent with respect to
$\langle\cdot,\cdot\rangle_{HS}$:
\begin{equation}
\label{eq:dephasing-selfadjoint-idempotent}
\mathscr{E}^2=\mathscr{E},
\qquad
\langle \boldsymbol{\hat{A}},\mathscr{E}(\boldsymbol{\hat{B}})\rangle_{HS}
=\langle \mathscr{E}(\boldsymbol{\hat{A}}),\boldsymbol{\hat{B}}\rangle_{HS}.
\end{equation}
Consequently, $\|\mathscr{E}(\boldsymbol{\hat{A}})\|_{HS}\le\|\boldsymbol{\hat{A}}\|_{HS}$.
\end{lemma}

\begin{proof}
Idempotence follows from $\mathscr{E}(\boldsymbol{\hat{A}})=\sum_k\boldsymbol{\hat{P}}_k\boldsymbol{\hat{A}}\boldsymbol{\hat{P}}_k$. Self-adjointness follows from cyclicity of the trace.
\end{proof}

\begin{proposition}[Coherence equals the squared HS norm of the off-diagonal part]
\label{prop:coherence-offdiagonal-norm}
Let $\boldsymbol{\hat{\sigma}}:=\mathscr{E}(\boldsymbol{\hat{\rho}})$ and
$\boldsymbol{\hat{\tau}}:=\boldsymbol{\hat{\rho}}-\boldsymbol{\hat{\sigma}}$. Then
\begin{equation}
\label{eq:coherence-equals-offdiagonal-hs}
C_{\mathscr{E}}(\boldsymbol{\hat{\rho}})=\|\boldsymbol{\hat{\tau}}\|_{HS}^2.
\end{equation}
\end{proposition}

\begin{proof}
$C_{\mathscr{E}}(\boldsymbol{\hat{\rho}})=\Tr(\boldsymbol{\hat{\rho}}^2)-\Tr(\boldsymbol{\hat{\sigma}}^2)=2\Tr(\boldsymbol{\hat{\sigma}}\boldsymbol{\hat{\tau}})+\Tr(\boldsymbol{\hat{\tau}}^2)=\Tr(\boldsymbol{\hat{\tau}}^2)$,
since $\boldsymbol{\hat{\sigma}}$ is diagonal and $\boldsymbol{\hat{\tau}}$ has vanishing diagonal
entries. Since $\boldsymbol{\hat{\tau}}$ is Hermitian, $\Tr(\boldsymbol{\hat{\tau}}^2)=\|\boldsymbol{\hat{\tau}}\|_{HS}^2$.
\end{proof}

\subsection*{B.2. A differential inequality for coherence growth under unitary dynamics}

Consider a time-dependent Hamiltonian $\boldsymbol{\hat{H}}(t)=\boldsymbol{\hat{H}}(t)^\dagger$ and
\begin{equation}
\label{eq:unitary-evolution-appendix}
\boldsymbol{\dot{\hat{\rho}}}(t)=-i[\boldsymbol{\hat{H}}(t),\boldsymbol{\hat{\rho}}(t)],
\qquad
\boldsymbol{\hat{\rho}}(0)=\boldsymbol{\hat{\rho}}_0.
\end{equation}

\begin{lemma}[Commutator bound]
\label{lem:commutator-hs-bound}
For any Hermitian $\boldsymbol{\hat{H}}$ and density operator $\boldsymbol{\hat{\rho}}$,
\begin{equation}
\label{eq:commutator-bound-density}
\|[\boldsymbol{\hat{H}},\boldsymbol{\hat{\rho}}]\|_{HS}\le 2\,\|\boldsymbol{\hat{H}}\|_{HS}.
\end{equation}
\end{lemma}

\begin{proof}
$\|[\boldsymbol{\hat{H}},\boldsymbol{\hat{\rho}}]\|_{HS}\le\|\boldsymbol{\hat{H}}\boldsymbol{\hat{\rho}}\|_{HS}+\|\boldsymbol{\hat{\rho}}\boldsymbol{\hat{H}}\|_{HS}\le2\|\boldsymbol{\hat{H}}\|_{HS}\|\boldsymbol{\hat{\rho}}\|_{op}\le2\|\boldsymbol{\hat{H}}\|_{HS}$.
\end{proof}

\begin{proposition}[Coherence growth inequality]
\label{prop:coherence-growth-ineq}
Along the evolution \eqref{eq:unitary-evolution-appendix},
\begin{equation}
\label{eq:integrated-coherence-growth}
\sqrt{C_{\mathscr{E}}(\boldsymbol{\hat{\rho}}(t))}-\sqrt{C_{\mathscr{E}}(\boldsymbol{\hat{\rho}}_0)}
\le
2\int_0^t \|\boldsymbol{\hat{H}}(s)\|_{HS}\,ds.
\end{equation}
\end{proposition}

\begin{proof}
By Proposition~\ref{prop:coherence-offdiagonal-norm},
$C_{\mathscr{E}}(\boldsymbol{\hat{\rho}}(t))=\|\boldsymbol{\hat{\tau}}(t)\|_{HS}^2$. Differentiating and applying Cauchy--Schwarz gives
$\frac{d}{dt}\|\boldsymbol{\hat{\tau}}(t)\|_{HS}\le\|\boldsymbol{\dot{\hat{\rho}}}(t)\|_{HS}\le2\|\boldsymbol{\hat{H}}(t)\|_{HS}$.
Integrating yields \eqref{eq:integrated-coherence-growth}.
\end{proof}

\subsection*{B.3. Proof of the coherence lower bound}

For time-independent generators,
$\mathscr{G}_{HS}(\boldsymbol{\hat{U}}(t))=\frac{t}{\sqrt{d}}\|\boldsymbol{\hat{H}}\|_{HS}$
where $d:=\dim(\mathfrak{su}(\dim(\mathscr{H})))$. For time-dependent generators,
\begin{equation}
\label{eq:hs-complexity-timedep}
\mathscr{G}_{HS}\!\bigl(\boldsymbol{\hat{U}}(t)\bigr)
\le
\frac{1}{\sqrt{d}}\int_{0}^{t}\|\boldsymbol{\hat{H}}(s)\|_{HS}\,ds,
\end{equation}
with equality when $\boldsymbol{\hat{U}}(t)$ is realized by the geodesic Hamiltonian path.

\begin{proposition}[Proof of Theorem~\ref{thm:coherence-lower-bound}]
Let $\boldsymbol{\hat{\rho}}(t)=\boldsymbol{\hat{U}}(t)\boldsymbol{\hat{\rho}}_0\boldsymbol{\hat{U}}(t)^\dagger$.
Then for all $t>0$,
\begin{equation}
\label{eq:coherence-lower-bound-general}
\int_0^t \|\boldsymbol{\hat{H}}(s)\|_{HS}\,ds
\ge
\frac{1}{4}\Bigl(C_{\mathscr{E}}(\boldsymbol{\hat{\rho}}(t))-C_{\mathscr{E}}(\boldsymbol{\hat{\rho}}_0)\Bigr).
\end{equation}
In particular, after taking the infimum over all $\boldsymbol{\hat{\rho}}_{0}$, the latter yields
\begin{equation}
\label{eq:coherence-lower-bound-complexity}
\mathscr G_{HS}\!\bigl(\boldsymbol{\hat U}(t)\bigr)
\ge
\frac{1}{4\sqrt{d}}\mathscr{C}_{\mathscr{E}}(\mathbf{\hat{U}}(t)).
\end{equation}
\end{proposition}

\begin{proof}
By Proposition~\ref{prop:coherence-offdiagonal-norm},
$C_{\mathscr{E}}(\boldsymbol{\hat{\rho}})=\|\boldsymbol{\hat{\tau}}\|_{HS}^2$.
Differentiating:
$\frac{d}{dt}C_{\mathscr{E}}(\boldsymbol{\hat{\rho}}(t))=2\Re\langle\boldsymbol{\hat{\tau}}(t),\boldsymbol{\dot{\hat{\tau}}}(t)\rangle_{HS}\le2\|\boldsymbol{\hat{\tau}}(t)\|_{HS}\|\boldsymbol{\dot{\hat{\tau}}}(t)\|_{HS}\le4\|\boldsymbol{\hat{H}}(t)\|_{HS}\|\boldsymbol{\hat{\tau}}(t)\|_{HS}$.
Dividing by $\|\boldsymbol{\hat{\tau}}(t)\|_{HS}$ (when nonzero) and integrating gives
$\sqrt{C_{\mathscr{E}}(\boldsymbol{\hat{\rho}}(t))}-\sqrt{C_{\mathscr{E}}(\boldsymbol{\hat{\rho}}_0)}\le2\int_0^t\|\boldsymbol{\hat{H}}(s)\|_{HS}\,ds$.
Since $\sqrt{a}-\sqrt{b}\ge(a-b)/(2\sqrt{a})$ for $a\ge b\ge0$ and $\sqrt{C_{\mathscr{E}}(\boldsymbol{\hat{\rho}}(t))}\le1$,
\eqref{eq:coherence-lower-bound-general} follows. Combining with \eqref{eq:hs-complexity-timedep}
yields \eqref{eq:coherence-lower-bound-complexity}, and the incoherent case is immediate.
\end{proof}

\begin{remark}[Interpretation]
The bound \eqref{eq:coherence-lower-bound-complexity} formalizes the intuition that, in Hilbert--Schmidt
geometry, generating off-diagonal weight in a fixed reference basis requires a minimum geometric
cost proportional to the coherence production, with universal normalization $(4\sqrt{d})^{-1}$.
\end{remark}

\section*{Declarations}

\subsection*{Conflict of interest}
The authors declare that they have no conflict of interest.

\subsection*{Ethical statement}
This study does not involve human participants, human data, or animals. Therefore, ethical approval was not required.

\subsection*{Informed consent}
Not applicable.

\subsection*{Data availability}
No datasets were generated or analyzed during the current study.

\subsection*{Funding}
This work was supported by the Generalitat Valenciana under grant COMCUANTICA/007 (QUANTWin), by the Agreement between the Directorate-General for Innovation of the Ministry of Innovation, Industry, Trade and Tourism of the Generalitat Valenciana and the Universidad CEU Cardenal Herrera, and by Universidad CEU Cardenal Herrera under grants INDI25/17 and GIR25/14.


\end{document}